\documentclass[11pt]{article}
\usepackage{macros}

\usepackage[citestyle=alphabetic,bibstyle=alphabetic,minalphanames=5,maxalphanames=6, maxbibnames=99]{biblatex}
\title{Random Reed-Solomon Codes and Random Linear Codes are Locally Equivalent}

\ifauthors
\author{Matan Levi\thanks{Ben-Gurion University. lmat@post.bgu.ac.il} \and Jonathan Mosheiff\thanks{Ben-Gurion University. mosheiff@bgu.ac.il. Supported by Israel Science Foundation grant 3450/24 and an Alon Fellowship. Part of this research was conducted while JM was visiting the Simons Institute for the Theory of Computing.}\and  
	Nikhil Shagrithaya\thanks{University of Michigan, Ann Arbor. nshagri@umich.edu. Supported by NSF awards CCF-2236931 and CCF-2107345. Part of this research was conducted while NS was visiting Ben-Gurion University.}}
\fi
\date{}

\begin{document}

	\maketitle

	\thispagestyle{empty}	
	\begin{abstract}
       We establish an equivalence between two important random ensembles of linear codes: random linear codes (RLCs) and random Reed-Solomon (RS) codes. Specifically, we show that these models exhibit identical behavior with respect to key combinatorial properties---such as list-decodability and list-recoverability---when the alphabet size is sufficiently large.

        We introduce monotone-decreasing local coordinate-wise linear (LCL) properties, a new class of properties tailored for the large alphabet regime. This class encompasses list-decodability, list-recoverability, and their average-weight variants. We develop a framework for analyzing these properties and prove a threshold theorem for RLCs: for any LCL property $\mathcal{P}$, there exists a threshold rate $R_\mathcal{P}$ such that RLCs are likely to satisfy $\mathcal{P}$ when $R < R_\mathcal{P}$ and unlikely to do so when $R > R_\mathcal{P}$. We extend this threshold theorem to random RS codes and show that they share the same threshold $ R_\mathcal{P} $, thereby establishing the equivalence between the two ensembles and enabling a unified analysis of list-recoverability and related properties.
        
        Applying our framework, we compute the threshold rate for list-decodability, proving that both random RS codes and RLCs achieve the generalized Singleton bound. This recovers a recent result of Alrabiah, Guruswami, and Li (2023) via elementary methods. Additionally, we prove an upper bound on the list-recoverability threshold and conjecture that this bound is tight.
        Our approach suggests a plausible pathway for proving this conjecture and thereby pinpointing the list-recoverability parameters of both models. Indeed, following the release of a prior version of this paper, Li and Shagrithaya (2025) used our equivalence theorem to show that random RS codes are near-optimally list-recoverable.
	\end{abstract}

	\clearpage
	\setcounter{tocdepth}{2}
    {\small \tableofcontents}
    
	\thispagestyle{empty}
	\clearpage
	\pagestyle{plain}
	\setcounter{page}{1}

	\section{Introduction} \label{sec:intro}
	An \deffont{(error-correcting) code} of length $n$ is a set $\cC$ of words over some fixed alphabet $\Sigma$. When $\Sigma$ is the finite field $\F_q$ (for some prime power $q$) and $\cC$ is a linear subspace of $\F_q^n$, we say that $\cC$ is a \deffont{linear code}. The \deffont{rate} of $\cC$ is $\frac{\log_{|\Sigma|}|C|}n$, which equals $\frac{\dim \cC}{n}$ when $\cC$ is linear. The \deffont{distance} of a linear code $\cC$ is $\min_{x\in \cC\setminus \inset 0}\wt x$, where $\wt x$ is the \deffont{Hamming weight}, which is the number of non-zero entries in $x$.

	In this work, we uncover a deep connection between two of the most important random ensembles of linear codes. These ensembles are:
	\begin{enumerate}
		\item The \deffont{random linear code (RLC)}. An \deffont{RLC} of length $n$ and rate $R$ is the kernel $\cC$ of a uniformly random matrix $P\in \F_q^{(n-k)\times n}$ where $k = Rn$.\footnote{A more common definition for an RLC is to sample a linear subspace of $\F_q^n$ at uniform from among all subspaces of dimension $k$. Fortunately, the two models are almost identical. See \cref{sec:Models} for details.}
		\item The \deffont{random Reed-Solomon (RS) code}. A \deffont{Reed-Solomon code} is defined by a length $n$, a dimension $1\le k\le n$, an alphabet $\F_q$, and a set of evaluation points $\alpha_1,\dots,\alpha_n\in \F_q$. Every codeword is the sequence of evaluations of some polynomial of degree $<k$ over the evaluation points.\footnote{Note that we allows evaluation points to repeat. There exists another model, arguably more common, in which the evaluation points are chosen without repetition. In \cref{sec:Models} we show that the two models are rather close, in a certain well-defined manner, allowing us to use one as a proxy for the other.} Formally, the code is
		$$\CRS[\F_q]{\inparen{\alpha_1,\dots,\alpha_n}}k := \inset{\inparen{f(\alpha_1),\dots,f(\alpha_n)}\mid f\in \F_q\inbrak{X},~\deg f < k}\eperiod$$
		\sloppy When $\alpha_1,\dots,\alpha_n$ are sampled independently and uniformly from $\F_q$, we say that $\CRS[\F_q]{\inparen{\alpha_1,\dots,\alpha_n}}k$ is a \deffont{Random RS Code}. 
	\end{enumerate}

	Our research is motivated by the desire to study two ubiquitous combinatorial properties of codes---\deffont{list-decodability} and \deffont{list-recoverability}. We define these notions as follows: let $\rho\in [0,1]$ and $L\in \N$. A code $\cC\subseteq \F_q^n$ is \deffont{$(\rho,L)$-list-decodable} if for every $z\in \F_q^n$, 
    \[
        \inabset{x\in \cC\mid \wt{x-z}\le \rho n}\le L,
    \]
    is true. Namely, every Hamming ball of radius $\rho n$ in $\F_q^n$ contains at most $L$ codewords. Fixing $\ell\in \N$, the code $\cC$ is said to be \deffont{$(\rho,\ell,L)$-list-recoverable} if for every list of sets $Z_1,\dots,Z_n\subseteq \F_q$, each of size at most $\ell$, there holds $$\inabset{x\in \cC\mid \inabset{i\in [n] \mid x_i\notin Z_i} \le \rho n}\le L\eperiod$$ 
	We note that $(\rho,L)$-list-decodability is equivalent to $(\rho,1,L)$-list-recoverability, so the latter notion generalizes the former. All of our results mentioned in this section also apply to the average-weight versions of list-decodability and list-recoverability (see \cref{sec:averageWeight}).

    \subsection{Our Contributions}\label{sec:contributions}
    In the first part of this work, we develop a new framework for analyzing a broad and natural family of properties, referred to as \deffont{(monotone-decreasing) local coordinate-wise linear (LCL) properties}.\footnote{In the main body of the paper (starting from \cref{sec:LCL}), we define and use the dual notion of a \deffont{(monotone-increasing) local coordinate-wise linear property}, which allows us to discuss the complements of properties like list-decodability and list-recoverability directly.} Notably,  list-decodability, list-recoverability and their average-weight variants fall into this category. For a given LCL property $\cP$, we examine the likelihood that an RLC $\cC$ of rate $R$ satisfies $\cP$. Our main result is a \deffont{threshold theorem}: we identify a specific \deffont{threshold rate} $R_\cP$ such that $\cC$ is highly likely to satisfy $\cP$ when $R < R_{\cP}$, and highly unlikely to satisfy it when $R > R_{\cP}$.

    Using our characterization of $R_\cP$, we compute the threshold rate for list-decodability (and its average-weight variant), showing that RLCs approach the generalized Singleton bound \cite{GST2022}, provided that the alphabet size $q$ is at least $2^{\Omega\inparen{L^2}}$. This recovers a recent result from \cite{AGL2023}, via an elementary approach that does not rely on external results such as the GM-MDS theorem \cite{DSY14, Lovett18, YH19} employed in \cite{AGL2023}. Additionally, we show that the threshold rate for $(\rho,\ell,L)$-list-recoverability is at most $\inparen{1-\frac 1{\log_\ell (L+1)}}\inparen{1-\frac{L+1}{L+1-\ell}\cdot \rho} + o(1)$ and conjecture this bound to be tight.\footnote{Since our paper was first published, \cite{LS25} have shown that the same negative bound, $R\le \inparen{1-\frac 1{\log_\ell (L+1)}}\inparen{1-\frac{L+1}{L+1-\ell}\cdot \rho} + o(1)$, holds for every linear code.}
    This upper bound was shown in a recent prior work by \cite{CZ24} for (Folded) Reed-Solomon codes, who also conjectured that this bound is tight for that code model. 

    The second part of this paper focuses on \deffont{random RS codes}. Our main result in this part is a threshold theorem for LCL properties in the random RS model. We show that a threshold exists for any such property $\cP$, and remarkably, this threshold is equal to $R_{\cP}$! This means that $\cP$ is likely to hold for a random RS code if and only if it is likely to hold for an RLC of similar rate, provided the alphabet size $q$ is sufficiently large. We thus say that these two code models are \deffont{equivalent for LCL properties} or \deffont{locally equivalent}.

    An immediate conclusion of this equivalence is that random RS codes and RLCs share the same list-recoverability parameters. For instance, to prove that a random RS code of rate $R$ is $\LR \rho \ell L$, it essentially suffices to establish the same for an RLC, and vice versa. Moreover, our characterization makes it possible to compute the relevant threshold $R_\cP$, yielding results applicable to both code models simultaneously.

    Prior works such as \cite{BGM2023, GZ2023, AGL2023} have proved that random RS codes also approach the generalized Singleton bound by using a technique first introduced in \cite{BGM2023}, which made use of an equivalence between the GMS-MDS theorem and list-decoding for random RS codes. Combining the list-decodability result for RLCs from the first part of our work with the equivalence result between random RS codes and RLCs in the second part, we show that random RS codes also approach the generalized Singleton bound. This effectively recovers the main result of the prior works without relying on the GM-MDS theorem.

    Our result about the local equivalence between random RS codes and RLCs, besides being interesting in its own right, provides a definite answer to the connection between the two well-studied code models, which had been hinted at in previous works. Furthermore, our result enables researchers to establish findings regarding local properties in one code model and directly transfer them to the other, allowing them to work with their preferred code model.

    \paragraph{Subsequent work.} The results of our work have already been used in \cite{LS25} to establish near-optimal upper bounds on the list size for list-recoverability of random RS codes. Specifically, they showed that random RS codes of rate $R$ are $(1 - R - \varepsilon, \ell, (\ell/\varepsilon)^{O(\ell/\varepsilon)})$-list-recoverable. Their approach proceeds by proving the bound for RLCs, after which the same result holds for random RS codes via our equivalence theorem. Notably, this is the first result on list-recovery of random RS codes that provides reasonable parameters for all rates $R$, whereas previous results were limited to the low-rate regime.

    \subsection{Local Code Properties and the LCL Framework}\label{sec:LCLIntro}
    Our notion of LCL properties extends a general framework, initially introduced in \cite{MRR+2020}, for studying list-decodability and list-recoverability as part of a broader class of code properties. A set of words in \( \mathbb{F}_q^n \) that all lie within the same Hamming ball of radius \( \rho n \) is called \deffont{\(\PC{\rho}\)}. A set of words \( X \subseteq \mathbb{F}_q^n \) is said to be \deffont{\(\RC \rho \ell\)} if there exists a sequence of sets \( Z_1, \dots, Z_n \subseteq \mathbb{F}_q \), each of size at most \( \ell \), such that for every \( x \in X \), the condition \( \left| \{ i \in [n] \mid x_i \notin Z_i \} \right| \le \rho n \) holds. Thus, a code \( \mathcal{C} \) is \(\LD \rho L\) (or \(\LR \rho \ell L\)) if and only if it does \emph{not} contain a \(\PC \rho\) (or \(\RC \rho \ell\)) set of size \( L+1 \). We interpret a clustered (or recovery-clustered) set of codewords in \( \mathcal{C} \) as a \deffont{witness} to the code's non-list-decodability (or non-list-recoverability). This leads to the observation that list-decodability and list-recoverability have short witnesses, drawing a natural analogy to \(\textsc{co-NP}\) languages or \(\Pi_1\) logical formulas.

    In \cite{MRR+2020}, this viewpoint was used to study list-decodability and list-recoverability within a newly defined class of \deffont{(monotone-decreasing) local code properties}. A property is said to be \deffont{local} if it has short witnesses, like clustered sets for list-decodability. The main finding of \cite{MRR+2020} is that \deffont{Gallagher codes}, a random ensemble of LDPC codes, possess local properties similar to those of RLCs, meaning they are just as list-decodable and list-recoverable (with high probability). Later, \cite{GM2022} extended this framework to show that \deffont{randomly punctured low-bias codes} also share local properties with RLCs. These results can be seen as \deffont{reductions} between code ensembles.\footnote{Note that \cite{MRR+2020,GM2022} only prove one-sided reductions, transferring monotone-decreasing local properties from RLCs to other ensembles. In contrast, the reduction in the present paper goes both ways.}

   The reduction results in \cite{MRR+2020} and \cite{GM2022} stem from the \deffont{threshold theorem for local properties for RLCs} \cite[Theorem 2.8]{MRR+2020}. This theorem states that for any local property $\cP$, there is a threshold rate $R_\cP$ such that RLCs with a rate below $R_\cP$ are likely to satisfy $\cP$, whereas those with a rate above $R_\cP$ are unlikely to do so.  An analogous theorem for plain random codes is proven in \cite{GMR+2022}.

    The threshold theorem crucially relies on the property $\cP$ having short witnesses. For instance, if $\cP$ is the property of $(\rho,L)$-list-decodability, the witnesses are $\PC \rho$ sets of size $L+1$. To apply the theorem, subsets of $\F_q^b$ (where $ b = L+1 $) are classified into types, and the theorem uses the expected number of these witnesses in a code $\cC$ to determine whether $\cC$ is likely to satisfy $\cP$. However, the number of types grows exponentially in $q^b$, which is manageable for fixed $q$ but becomes unwieldy when $q$ increases with $n$. For a more formal treatment of this technique, see \cite{MRSY24}.

    This exponential growth in the number of types motivates the need for a refined approach. Our framework of \deffont{LCL properties} avoids explicitly enumerating types and can handle cases where \( q \) is large, often keeping the complexity independent of \( q \) (see \cref{rem:reasonableProperties}).

      Notably, random RS codes arise as random puncturings of the \deffont{full RS code} \( \CRS[\mathbb{F}_q]{\mathbb{F}_q}{k} \), which is known to have low bias, suggesting that the results of \cite{GM2022} could apply to their analysis. However, the methods developed in \cite{MRR+2020,GM2022} are tailored to codes over relatively small alphabets and do not extend to RS codes, where \( q \ge n \) is required. Our LCL framework fills this gap by extending the local properties approach to the large-alphabet regime.
    
    Informally, a property $\cP$ is $b$-LCL if it is witnessed by a small set of codewords $x^1,\dots, x^b \in \F_q^n$ such that, for each $1\le i\le n$, the vector $(x^1_i, \dots, x^b_i)$ satisfies certain linear constraints. For example, $\rho$-clusteredness is characterized by a large number of equality constraints (which are, in particular, linear) among $x^1_i,\dots, x^b_i$ for many coordinates $1\le i\le n$. We provide a formal treatment in \cref{sec:LCL}.

    As mentioned, the LCL framework encompasses list-decodability and list-recoverability, along with their average-weight variants. It also naturally captures other properties, such as average pairwise distance (defined in \cite{CGV2013}) and list-decodability from burst errors (e.g., \cite{RV2009}) under a unified lens.        

    \subsection{Background \footnote{This section describes the state of knowledge as it stood prior to the initial publication of this work. It does not reflect the subsequent contribution of \cite{LS25}, discussed in \cref{sec:contributions}, which established a new result on the list-recoverability of RLCs in the large alphabet regime and, using our main theorem, significantly advanced the known bounds for list-recoverability of random RS codes.}
}

	\subsubsection{List-Decodability and List-Recoverability of RS Codes}\label{sec:prevRS}
	In light of the importance of RS codes, it is a major open problem to explicitly construct RS codes with good list-decodability or list-recoverability parameters. Currently, no non-trivially list-decodable or list-recoverable explicit RS codes are known. For list-decodability, this means that known explicit RS codes are only list-decodable up to the Johnson bound (\cite{Sudan1997}, \cite{GS1998}). A negative list-decodability result for the full RS code $\CRS[\F_q]{\F_q}k$ in some parameters regimes is given in  \cite{BKR09}.
	
	Lacking explicit constructions, much attention has been given to the corresponding existential problems. List-decodability of random RS codes was studied in \cite{RW2014,ST2020,GLS+2021,FKS2022,GST2023}. Finally, in \cite{BGM2023, GZ2023, AGL2023} it was shown that random RS codes are list-decodable up to capacity,\footnote{A family of codes is list-decodable up to capacity if it achieves $(\rho,L)$-list-decodability with rate $1-h_q(\rho)-\eps$ for arbitrarily small $\eps$ and $L \le \poly(n)$. For large $q$, the rate approaches $1-\rho-\eps$.} and, furthermore, that they achieve the \deffont{Elias bound}\footnote{A family of codes achieves the Elias bound if it satisfies $(\rho, L)$-list-decodability with rate $1 - h_q(\rho) - \epsilon$ for arbitrarily small $\epsilon$ and $L \le O\left(\frac{1}{\epsilon}\right)$. This is a stricter requirement than achieving list-decoding capacity. Essentially, codes achieving the Elias bound are at least as list-decodable as plain random codes. See \cite{MRSY24}.} and the \deffont{generalized Singleton bound} \cite{ST2020}. \cite{BGM2023} noted an interesting connection between the GM-MDS theorem (proven in \cite{Lovett18,YH2019}) and list-decoding of RS codes, which allowed them to prove that random RS code can achieve the Generalized Singleton bound, but with exponential field size. In \cite{BDG24} the authors prove that in order to achieve the Generalized Singleton bound exactly, exponential field size is in fact necessary. However, \cite{GZ2023} circumvent this by showing that it is possible to approach the bound arbitrarily closely with only quadratic field size, which was improved to linear by the work of \cite{AGL2023}. Recent works have generalized these results to Algebraic-Geometric codes \cite{BDGZ24}, other polynomial based code ensembles \cite{BDG23}, and also to Gabidulin codes in the rank metric \cite{GXYZ24}. We note that \cite{BGM2023,GZ2023,AGL2023,BDG23,BDGZ24,GXYZ24} all utilize a common framework that relies on the GM-MDS theorem, or a variant of it as a crucial component. It is not yet clear whether this framework can be extended to deal with other local properties, such as list-recovery.
	
	Much less is known about list-recovery of RS codes, even in the random setting. The existing results include \cite{LP2020,GLS+21,GST2023}. The first is a non-trivial list-recovery result for random RS codes, where the rate decreases to zero as $n$ grows, while the second shows the existence of random RS codes of rate $\Omega\left(\frac{\varepsilon}{\sqrt{\ell}\cdot \log 1/\varepsilon}\right)$ that are $(1-\varepsilon, \ell, O(\ell/\varepsilon))$-list-recoverable. The third result shows that random RS codes are $\LR \rho \ell L$ with rate approaching $\frac{1-\rho}{\ell + \rho}$.

    We mention that for explicit Folded Reed-Solomon (FRS) codes, there have been a series of recent works \cite{KRS+2018, Tam24, Sri24, CZ24} which prove positive list-decodability and list-recoverability results for these codes. In particular, \cite{Tam24} proved that explicit FRS codes having relative distance $\delta$ are $(\delta-\varepsilon, \ell, L)$-list-recoverable with $L \leq O\left(\frac{\ell}{\varepsilon}^{O\left(\frac{1+\log \ell}{\varepsilon}\right)}\right)$. For the setting of list-decoding, \cite{CZ24} proved that the codes meet the generalized Singleton bound, and thus are $(1-R-\varepsilon, L)$-list-decodable with $L \leq O(1/\varepsilon)$.
	Additionally, they also prove a negative result for the list-recoverability of Folded Reed-Solomon codes, showing that every Folded Reed-Solomon code of rate $R$ having distinct evaluation points cannot be $(1-R-\varepsilon, \ell, \ell^{\frac{R}{2\varepsilon}-1}-1)$ list-recoverable. This result holds for every constant folding parameter $s \geq 1$, and hence also holds for Reed-Solomon codes having distinct evaluation points.
	
       \subsubsection{List-Decodability and List-Recoverability of RLCs}\label{sec:RLCPrev}

        The $(\rho, L)$-list-decodability of RLCs for a fixed $L$ has been the focus of extensive research \cite{ZP1981,GHS+2002,GHK2011,CGV2013,Wootters2013,RW2014,RW2018,LW2021,GLM+2022,AGL2023}, employing at least four different methods. Collectively, these works demonstrate that RLCs achieve the Elias bound in the following settings:
        \begin{itemize}
            \item When $q = 2$ \cite{GHS+2002, LW2021}.
            \item When $q \le O(1)$ and $\rho$ is bounded away from $\frac{q-1}{q}$ \cite{GHK2011}.
            \item When $q \ge 2^{\Omega\left(\frac{L}{\epsilon}\right)}$, as shown in \cite{AGL2023} and in this work.
        \end{itemize}
        The studies \cite{CGV2013, Wootters2013, RW2014} provide positive results in the regime where $\rho \to \frac{q-1}{q}$, though these do not achieve the Elias bound. A significant open problem remains: proving that RLCs meet the Elias bound across all parameter settings, ideally through a unified proof technique.
        
        The study of list-recoverability for RLCs can be divided into the large $q$ and small $q$ regimes. For large $q$ (e.g., when $q$ is exponentially large in $L$), the only known positive result comes from a straightforward application of the Zyablov-Pinsker Lemma \cite{ZP1981}, which provides a weak lower bound on the threshold rate that rapidly diminishes as $q$ or $L$ increase. In this work, we establish the upper bound $R_{\cP} \le \left(1 - \frac{1}{\log_\ell(L+1)}\right)\left(1 - \frac{L+1}{L+1-\ell} \cdot \rho\right)$ and  conjecture that this bound is tight.
        This is the same bound that was first proven in a recent prior work by \cite{CZ24} for $(\rho, \ell, L)$-list-recoverability of (Folded) Reed-Solomon codes. Importantly, our result indicates that RLCs do not reach an analogue of the Elias bound for list-recovery in the large $q$ regime, meaning they are not as list-recoverable as plain random codes.
        
        List-recovery in the regime where $\ell \ge q^{\Omega(1)}$ is studied in \cite{RW2014, RW2018}, which give positive results. However, \cite{GLM+2022} demonstrates that when $q$ is a large power of a small prime, RLCs underperform significantly compared to plain random codes in this setting as well.

	\subsection{Open Problems}
    \subsubsection{List-Recovery and Non-Local Properties}

        Determining the threshold for list-recovery in the large $q$ regime remains a significant open problem, as it would reveal the list-recovery parameters of both random RS codes and RLCs. As noted earlier, we conjecture this threshold to be $\left(1 - \frac{1}{\log_\ell(L+1)}\right)\left(1 - \frac{L+1}{L+1-\ell} \cdot \rho\right)$.
       
        Another challenge arises in the small $q$ regime, such as when $q = \mathrm{poly}(\ell)$. Our current methods cannot directly address list-recoverability in this setting due to the exponential dependence of $q$ on locality. Relaxing this dependence is an important goal, as it could deepen our understanding of list-recovery in this regime.
        
        Furthermore, extending our framework to analyze highly non-local properties would be valuable. Consider the following example: fix a large $q$ and let $0 < R < 1$, $q^{(1-R)} < \ell < q$, and $\epsilon > 0$. Let $\cP$ represent the property of $\inparen{0, \ell, L}$-list-recoverability for codes in $\F_q^n$, where $L = \ell^n \cdot q^{-(1-R)n} \cdot (1 + \epsilon)$. Note that $\cP$ is not a local property, since $L$ is typically exponential in $n$, but it is a natural one to consider.
        
        Notice that $\ell^n \cdot q^{-(1-R)n}$ is the expected size of the set $\cC \cap Z_1 \times Z_2 \times \dots \times Z_n$, where $\cC \subseteq \F_q^n$ is a fixed code of rate $(1-R)$, and $Z_1, \dots, Z_n$ are independent uniformly random subsets of $\F_q$ of size $\ell$. Thus, a code $\cC$ satisfies $\cP$ if its maximum intersection with any product set $Z_1 \times Z_2 \times \dots \times Z_n$ is at most $(1+\epsilon)$ times this expectation. From a pseudorandomness perspective, a code satisfying $\cP$ can be said to \deffont{fool product sets (also known as \deffont{combinatorial rectangles})} within a $(1+\epsilon)$ factor.

    Consider the question of whether any linear code satisfies this property. In particular, do RLCs satisfy it? Both questions were essentially answered in the affirmative by \cite{MPS+2021} when \( R > \frac{1}{2} \) and \( q \) is a large prime. For \( R \le \frac{1}{2} \), these questions remain open. We conjecture that the answer is again positive, provided that \( q \) is a sufficiently large prime.

        \subsubsection{An Alphabet-Uniform Framework}
        
        As discussed in \cref{rem:reasonableProperties}, the LCL framework developed here is well-suited for sufficiently large alphabets. For instance, \cref{cor:RLCLD} on the list-decodability of RLCs applies when $q \ge 2^{\Omega\left(\frac{L}{\epsilon}\right)}$, where $\epsilon$ is the gap to capacity. In contrast, the classical local property framework from \cite{MRR+2020} is designed for smaller alphabets. Our goal is to unify these approaches into a single framework that works uniformly across all alphabet sizes.
        
        Currently, RLCs are shown to achieve the Elias bound for list-decodability in two cases: when $q$ is constant or when $q \ge 2^{\Omega\left(\frac{L}{\epsilon}\right)}$, leaving a gap between these two regimes. This gap highlights the differences between the two methods used to study local properties. The proof for the constant $q$ case \cite{GHK2011} aligns well with the classical local properties framework, whereas our proof for the large $q$ case (\cref{cor:RLCLD}) is captured within the LCL framework. Developing a unified local property framework may provide a uniform proof that RLCs achieve the Elias bound across all alphabet sizes.

	\subsubsection{Optimality of Random Linear Codes}\label{sec:prevOptimality}
        In the large \( q \) regime, RLCs and random RS codes are essentially optimally list-decodable, as they approach the generalized Singleton bound. We conjecture that this optimality extends more broadly within the class of linear codes and holds for all LCL properties (see \cref{sec:LCL} for a formal definition).

	\begin{conjecture}[Optimality of RLCs and random RS codes for LCL properties]\label{conj:optimal}
		For every $\eps > 0$ and $b \in \N$, there exist constants 
		$n_0$ and $q_0$ such that the following holds: Let $\cP$ be 
		a $b$-local (monotone-decreasing) LCL property. If $q \ge q_0$ 
		and $n \ge n_0$, then any linear code in $\F_q^n$ with rate 
		at least $R_\cP + \eps$ does not satisfy $\cP$.
	\end{conjecture}

    Following the original publication of this paper, \cite{LS25} essentially proved the conjecture for the case of list-recoverability, by proving nearly tight upper bounds on the output list size (which is equivalent to proving an upper bound on the threshol rate corresponding to a fixed output list size).
    If the conjecture is true, then it would simultaneously generalize this negative result for list-recovery and the generalized Singleton bound.

	\section{Preliminaries} \label{sec:preliminaries}
	\subsection{General Notation}
	We denote the set $\{1,\ldots,n\}$ by $[n]$ and let $\F_q$ be the finite field of order $q$, where $q$ is a prime power.  Given a vector space $V$, the family of all linear subspaces of $V$ is denoted by $\cL(V)$.

	Given a matrix $A\in \F_q^{n\times b}$ ($b\in \N$) and a code $\cC\subseteq \F_q^n$, we write $A\subseteq \cC$ to mean that every column of $A$ belongs to $\cC$. We then say that $\cC$ \deffont{contains} $A$. We write $\mrow Ai$ and $\mcol Aj$ to refer, respectively, to the $i$-th row and $j$-th column of $A$. We denote $$\mdist qnb = \inset{A\in \F_q^{n\times b} \mid A\textrm{ has pairwise-distinct columns}}\eperiod$$
	
	We use boldface symbols $\0, \1$ to represent the all zeroes vector and the all ones vector, respectively. If $I\subseteq [b]$, we let $\1_I\in \F_q^b$ be the indicator vector for $I$.
	
	\subsection{Average-Weight List-Decodability and Average-Weight List-Recoverability}\label{sec:averageWeight}
	As mentioned in \cref{sec:intro}, list-decodability and list-recoverability have average-weight variants. We define them here.
	
	\begin{definition}
		Let $X\subseteq \F_q^n$ be a set of words. If there exists some $z\in \F_q^n$ such that $\frac{\sum_{x\in X} \wt{x-z}}{|X|} \le \rho n$, we say that $X$ is \deffont{$\APC \rho$}. If there exists a sequence of sets $Z_1,\dots, Z_n\subseteq \F_q^n$, each of size at most $\ell$, such that $\frac{\sum_{x\in X} \inabset{i\in [n]\mid x_i\notin Z_i}}{|X|}\le \rho n$, we say that $X$ is \linebreak \deffont{$\ARC \rho \ell$}.
		
		A code is said to be \deffont{$\ALD \rho L$} if it does not contain a \linebreak $\APC \rho$ set of size larger than $L$. A code is said to be \linebreak \deffont{$\ALR \rho \ell L$} if it does not contain a $\ARC \rho \ell$ set of size larger than $L$.
	\end{definition}
	
	Note that the average-weight variants are stronger than the plain versions of these properties. Specifically, if a code is $\ALD{\rho}{L}$, it is also $\LD{\rho}{L}$, and similarly for list-recoverability.
	
	\subsection{Local Coordinate-Wise Linear Properties of Codes}\label{sec:LCL}
	A \deffont{property} of ($q$-ary length $n$) codes is a family of codes $\cP$ in $\F_q^n$. For a code $\cC\subseteq \F_q^n$, if $\cC\in \cP$, we say that $\cC$ \deffont{satisfies} $\cP$. 
	
	We now define a special class of code properties called \deffont{local coordinate-wise linear properties}. Fix $b\in \N$. A sequence $\cV = \inparen{\cV_1,\dots, \cV_n}\in \cL(\F_q^b)^n$ is called a \deffont{$b$-local profile}. A matrix $A\in \F_q^{n\times b}$ is said to \deffont{satisfy} the profile $\cV$ if $\mrow Ai\in \cV_i$ for all $1\le i \le n$. We write 
	$$\cM_{\cV} = \inset{A\in \F_q^{n\times b}\mid A\text{ satisfies }\cV}$$
	and denote
	$$\cM^{\distinct}_{\cV} = \cM_{\cV}\cap \mdist qnb\eperiod$$
	
	A code $\cC\subseteq \F_q^n$ is said to \deffont{contain} $\cV$ if it contains some matrix $A\in \cM^{\distinct}_\cV$. A property $\cP$ of length $n$ codes is \deffont{(monotone-increasing)\footnote{The properties discussed informally in \cref{sec:intro} are monotone-decreasing, namely, adding codewords to a code $\cC$ makes them harder to satisfy. List-decodability and list-recoverability are monotone-decreasing properties. Henceforth, we deal with monotone-increasing properties instead. In particular, rather than directly studying list-decodability and list-recoverability, we investigate the complements of these properties.} $b$-LCL ($b$-local coordinate-wise linear)} if there is a family of $b$-local profiles $\cF \subseteq \cL(\F_q^b)^n$ such that
	$$\cP = \inset{C\subseteq \F_q^n \mid \exists \cV\in \cF \text{ such that }C\text{ contains }\cV}\eperiod$$
	In other words, the property $\cP$ is $b$-LCL if it consists of those codes that satisfy at least one profile from a certain family of $b$-local profiles.
    
    The notion of $b$-local profiles can be regarded as a generalization of Intersection Matrices, originally introduced in \cite{ST2020} (see also \cite{GZ2023}, which presents an equivalent concept under the name Reduced Intersection Matrices). The fundamental observation underlying these works is that any code that is not list-decodable must contain a subset of codewords that exhibit significant agreement with some fixed vector in the ambient space. Equivalently, these codewords must agree with one another on a substantial number of coordinates. \cite{ST2020} and \cite{GZ2023} leveraged this insight by representing codewords as columns of a matrix, where coordinate-wise agreements among codewords manifest as linear constraints on the matrix rows. (Reduced) Intersection Matrices serve as symbolic matrices that encode such agreement constraints. The concept of a $b$-local profile extends this framework by encoding constraints on coordinates in a more general manner, permitting the representation of linear constraints rather than merely agreement constraints. This generalization plays a crucial role in establishing the threshold theorem for RLCs and in demonstrating the equivalence between RLCs and random RS codes.
    
    As shown by the following proposition, this framework allows us to capture natural code properties such as list-decodability and list-recoverability.
	\begin{proposition}\label{prop:LDLRareLDL}
		The following holds:
		\begin{enumerate}
			\item  The complement of $(\rho, L)$-list-decodability is an $(L+1)$-LCL property with an associated $(L+1)$-local profile family of size at most $\binom{n}{\rho n}^{L+1}$.
			\item  The complement of $(\rho, \ell,L)$-list-recoverability is an $(L+1)$-LCL property with an associated $(L+1)$-local profile family of size at most $\binom{n}{\rho n}^{L+1}\cdot \ell^{(L+1)n}$.
		\end{enumerate}
	\end{proposition}
	\begin{proof}
		It suffices to prove the claim for $(\rho,\ell,L)$-list-recoverability, since list-decodability is merely list-recoverability with $\ell=1$. Consider the code property
		$$\cP = \inset{\cC\subseteq \F_q^n \mid \cC \text{ is not }(\rho,\ell,L)\text{-list-recoverable}}\eperiod$$
		To prove that $\cP$ is $(L+1)$-LCL, we define a corresponding set $\cF$ of $(L+1)$-local profiles.
		
		Let $I_1,\dots, I_{L+1}\subseteq [n]$ be sets, each of size at least $(1-\rho)\cdot n$. For each $i\in [n]$, let $\sim_i$ be an equivalence relation over $[L+1]$, consisting of at most $\ell$ equivalence sets. Denote $\cI = (I_1,\dots, I_{L+1})$ and $\sim =(\sim_1,\dots,\sim_n)$.  Define the profile $\cV^{\cI,\sim}\in \cL\inparen{\F^{L+1}_q}^n$ by
		\[
		\cV^{\cI,\sim}_i := \inset{x \in \F_q^{L+1} \mid \forall r,s \in [L+1]\textrm { if }i\in I_r\cap I_s\textrm{ and }r\sim_i s\textrm{ then }x_r=x_s} \eperiod
		\]
		Let $\cF$ be the family of all profiles $\cV^{\cI,\sim}$ for $\cI$ and $\sim$ of the above form.  Observe that $|\cF| \le \binom{n}{\rho n}^{L+1}\cdot \ell^{(L+1)n}$.
		
		To prove the proposition it suffices to show that
		$$\cP = \inset{\cC\subseteq \F_q^n\mid \exists \cV\in \cF \text{ such that $\cC$ contains $\cV$}}\eperiod$$
		We do so by proving containment in both directions.
		
		First, suppose that $\cC$ contains some $\cV^{\cI,\sim}\in \cF$, where $\cI$ and $\sim$ are as above. Let $y^1,\dots, y^{L+1}\in \cC$ be distinct codewords that satisfy $\cV$. Define the sets $Z_1,\dots, Z_n\subseteq \F_q$ by $Z_i = \inset{y^r_i \mid i\in I_r}$. It is straightforward to verify that, because $y^1,\dots,y^{L+1}$ satisfies $\cV^{\cI,\sim}$, it must hold that $|Z_i|\le \ell$. Furthermore, since each set $I_r$ is of size at least $(1-\rho)n$, there are at most $\rho n$ coordinates $i$ in which $y^r_i\notin Z_i$. Thus, $y^1,\dots,y^{L+1}$ is a witness that $\cC$ is not $(\rho,\ell,L)$-list-recoverable.
		
		In the other direction, suppose that $\cC$ is not $(\rho,\ell,L)$-list-recoverable. Let $y^1,\dots,y^{L+1}$ be a witness to this fact and let $Z_1,\dots, Z_{n}$ be a corresponding sequence of input lists, each of size at most $\ell$. For $r\in [L+1]$, let $I_r = \inset{i\in [n] \mid y^r_i \in Z_i}$. Clearly, $|I_r|\ge (1-\rho) n$. Let $\sim_i$ be an equivalence relation over $[L+1]$ such that whenever $y^r_i = y^s_i$ and $y_r\in Z_i$, then $r\sim_i s$. Note that there exists such a relation with at most $|Z_i|\le \ell$ equivalence sets. It is now straightforward to verify that $y^1,\dots, y^{L+1}$ satisfy the profile $\cV^{\cI,\sim}$ where $\cI = (I_1,\dots, I_{L+1})$ and $\sim = (\sim_1,\dots,\sim_n)$.
	\end{proof}
	
	\section{Organization and Formal Statements of Main Results}
	Having defined the notion of \deffont{LCL properties}, we turn to formally state our main results. 
	
	\subsection{Results for Random Linear Codes}
	In \cref{sec:RLCLCL} we study LCL properties of RLCs, and prove the \deffont{threshold theorem} for RLCs.
	\begin{theorem}[RLC thresholds for LCL properties over a large alphabet]\label{thm:RLCThresholdForLCLIntro}
		Let $\cP$ be a $b$-LCL property of codes in $\F_q^n$ and let $\cF\subseteq \cL\inparen{\F_q^b}^n$ be a corresponding family of profiles. Let $\cC\subseteq \F_q^n$ be an RLC of rate R. Then, there is some threshold rate $R_\cP$ for which the following holds.
		\begin{enumerate}
			\item If $R \ge R_{\cP}+\eps$ then $\PR{\cC\textrm{ satisfies }\cP} \ge 1 - q^{-\eps n+b^2}$.
			\item If $R \le R_{\cP}-\eps$ then $\PR{\cC\textrm{ satisfies }\cP} \le \inabs{\cF}\cdot q^{-\eps n + b^2}$.
			\item In particular, if $R \le R_{\cP}-\eps$ and $q \ge 2^{\frac{2\log_2 |\cF|}{\eps n}}$ then $\PR{\cC\textrm{ satisfies }\cP} \le q^{-\frac{\eps n}2 + b^2}$.
		\end{enumerate}
	\end{theorem}
	\begin{remark}[Characterization of $R_\cP$]\label{rem:ThresholdCharacterization}
		The usefulness of \cref{thm:RLCThresholdForLCLIntro} depends on having a clear characterization of $R_\cP$ in terms of certain first-moment terms. Such a characterization is developed and given explicitly in \cref{eq:R_PDef} in \cref{sec:RLCLCL}.
	\end{remark}
	\begin{remark}[The alphabet size and ``reasonable'' properties]\label{rem:reasonableProperties}
		As demonstrated by the second and third parts of \cref{thm:RLCThresholdForLCLIntro}, the usefulness of the theorem hinges on $\cF$ not being too large in terms of $q$. For the probability bound to be meaningful, we need 
		\begin{equation}\label{eq:FInequality}
			q \ge |\cF|^{\Omega{\left(\frac{1}{\eps n}\right)}}\ecomma
		\end{equation} at the very least. It seems that many natural LCL properties have $|\cF|\ge 2^{\Omega(n)}$ at the very least, making $q \ge 2^{\Omega\inparen{\frac 1\eps}}$ a minimum requirement.
		
		Informally, we define a property $\cP$ as \deffont{reasonable} if its associated set of local profiles satisfies 
		\begin{equation}|\label{eq:FSmall}
			\cF| \le q^{o_{q\to \infty}(n)}\eperiod
		\end{equation}
		Note that only a reasonable property can satisfy \cref{eq:FInequality}.
		Fortunately, the complements of list-recoverability and list-decodability with fixed list-size are reasonable (\cref{prop:LDLRareLDL}).
		
		Furthermore, let $\cP$ denote a $b$-LCL property and let $M = \inset{\cV_i \mid \cV\in \cF}$, the set of all linear subspaces of $\F_q^b$ that can appear in a profile in $\cF$. It is not hard to see that $|\cF|\le |M|^n$. Hence, as long as $|M|$ depends only on $b$ but not on $q$, the property $\cP$ is reasonable. 
		
		For example, suppose that $\cP$ is the property of not being $\LR \rho \ell L$. As demonstrated in the proof of \cref{prop:LDLRareLDL}, each element of $M$ is defined by an equivalence relation over $[b]$, and thus, $|M|\le b^b$. Therefore, \cref{eq:FSmall} is satisfied, making $\cP$ reasonable.
		
		In the classic local property framework (see \cite{MRR+2020}), the number of \deffont{types} associated with a $b$-local property is equal to the number of ways to distribute $n$ unlabeled balls between $q^b$ labeled bins, a term which grows exponentially in $q$. In the present LCL framework, $|M|$ can be seen as analogous to the base $n$ logarithm of the number of types. The fact that $|M|$ can remain constant as $q$ grows is crucial to the suitability of LCL properties for studying the large alphabet regime.
	\end{remark}
	
	In \cref{sec:RLCLD} we use the characterization mentioned in \cref{rem:ThresholdCharacterization} to compute the threshold rate for (average-weight) list-decodability and to give an upper bound on the threshold for list-recoverability.
	
	\begin{restatable}[RLC threshold for list-decodability]{theorem}{RLCLD} \label{thm:RLCThresholdForLDIntro}
		\sloppy
		Fix $\rho \in [0,1]$ and $L\in \N$. Consider the properties $\cP := \inset{\cC\subseteq \F_q^n \mid \cC\textrm{ is not }\LD\rho L}$ and $\cP' := \inset{\cC\subseteq \F_q^n \mid \cC\textrm{ is not }\ALD\rho  L}$. Then, $$ R_{\cP} \ge R_{\cP'} \ge \max\inset{1-\rho\cdot\inparen{1+\frac 1 L},0}\eperiod$$
		Furthermore, if $n$ is divisible by $\binom{L+1}{(1-\rho)(L+1)}$ then the above bound is tight, namely,
		$$ R_{\cP} = R_{\cP'} = \max\inset{1-\rho\cdot\inparen{1+\frac 1 L},0}\eperiod$$
	\end{restatable}
	\begin{remark}        
		The hard part of this theorem is proving the lower bound on $R_\cP$. This is a positive result about RLCs, which also follows from \cite{AGL2023}. The reasoning in \cite{AGL2023} relies on that paper's main result about random RS codes, and, in particular, on the GM-MDS theorem. In contrast, our proof is more direct and elementary.
		
		Our upper bound on $R_\cP$ can also be inferred from the generalized Singleton bound \cite{GST2022} (see \cref{sec:prevOptimality}). Here we prove it directly (assuming the proper divisibility condition) within the LCL framework.
    \end{remark}
	
	\cref{thm:RLCThresholdForLCLIntro,thm:RLCThresholdForLDIntro,prop:LDLRareLDL} immediately yield the following corollary.

    \begin{corollary}[List-decodability of RLCs]\label{cor:RLCLD}
		Fix $\rho \in[0,1]$ and $L\in \N$. Let $\cC\subseteq \F_q^n$ be an RLC of rate $R$. Let $\eps > 0$. The following now holds.
		\begin{enumerate}
			\item If $R \le 1-\rho\cdot\inparen{1+\frac 1L} - \eps$ then $\cC$ is $\LD \rho L$ with probability at least $1-2^{n\cdot(L+1)}\cdot q^{-\eps n+(L+1)^2}$. 
			\item In particular, if $R \le 1-\rho\cdot\inparen{1+\frac 1L} - \eps$ and $q > 2^{\frac{2(L+1)}\eps}$ then $\cC$ is $\ALD \rho L$ with probability at least $1-q^{-\frac{\eps n}2+(L+1)^2}$.
		\end{enumerate}
	\end{corollary}
	\begin{remark}
		A very similar result to \cref{cor:RLCLD} is proved in \cite{AGL2023}.
	\end{remark}

    	\begin{restatable}[Upper bound on RLC threshold for list-recoverability]{theorem}{RLCLR} \label{thm:RLCThresholdForLRIntro}
        \sloppy
		Fix $\rho \in [0,1]$ and $\ell,L\in \N$ such that $\ell \ge 2$ and $L+1 = \ell ^m$ for some $m\in \N$. Consider the properties $\cP := \inset{\cC\subseteq \F_q^n \mid \cC\textrm{ is not }\LR\rho\ell L}$ and $\cP' := \inset{\cC\subseteq \F_q^n \mid \cC\textrm{ is not }\ALR\rho \ell L}$. 
       Suppose that $n\cdot\left( 1-\frac{L+1}{L+1-\ell}\cdot \rho\right)$ is divisible by $m$ and $(L+1)\rho n$ is divisible by $\binom{L+1}\ell$. Then,
       $$ R_{\cP'} \leq R_{\cP} \leq \left(1- \frac{1}{m}\right)\left( 1-\frac{L+1}{L+1-\ell}\cdot \rho\right)\eperiod$$ 
	\end{restatable}
    \begin{remark}
       The divisibility conditions in \cref{thm:RLCThresholdForLCLIntro,thm:RLCThresholdForLRIntro} are not very significant, as the theorems provide nearly identical bounds even if the length $n$ does not satisfy these conditions. Indeed, one can puncture the RLC $\cC$, removing a tiny fraction of its coordinates to obtain a shortened code $\cD$ that satisfies the divisibility condition, allowing the theorem to be applied. Moreover, $\cC$ and $\cD$ differ only slightly in their list-decodability and list-recoverability parameters, so the result for $\cD$ also gives a close bound for $\cC$.
    \end{remark}

	\subsection{The Equivalence Theorem and Results About Random Reed-Solomon Codes }
	In \cref{sec:RSLDL,sec:RSContainsPRofile} we prove the equivalence between  RLCs and random RS codes.
	\begin{restatable}[Threshold theorem for RS codes]{theorem}{reduction}\label{thm:MainRSIntro}
		Let $\cP$ be a $b$-LCL property of codes in $\F_q^n$, with associated local profile family $\cF\subseteq \cL\inparen{\F_q^b}^n$. Let $0 < R < 1$ and let $\cC = \CRS[\F_q]{\alpha_1,\dots,\alpha_n}k$, where $k = Rn$ and $\alpha_1,\dots,\alpha_n$ are sampled independently and uniformly from $\F_q$. Fix an $\eps>0$ satisfying  $\eps n \ge 2b(b+1)$. Furthermore, let $q$ satisfy $q > c (4b)^{4b}k/\eps$, for some constant $c>1$. The following now holds:
        \begin{enumerate}
            \item If $R\le R_\cP-\eps$, then 
		      \begin{equation}\label{eq:MainReductionRLCtoRS}
		          \PR{\cC \textrm{ satisfies }\cP} \le (2^b-1) \cdot \inparen{\frac {(4b)^{4b}k}  {\eps q}}^{\frac{\eps n}{2b}} \cdot |\cF|\eperiod
		      \end{equation}
            \item If $R\ge R_{\cP}+\eps$, then
            		      \begin{equation}\label{eq:MainReductionRStoRLC}
                                \PR{\cC \textrm{ satisfies }\cP} \ge 1 - \binom b2\cdot (2^b-1)\cdot \inparen{\frac {(4b)^{4b}k}{\eps q}}^{\frac{\eps n}{2b}} \eperiod
                            \end{equation}
        \end{enumerate}
	\end{restatable}
    For a fixed LCL property $\cP$, \cref{thm:MainRSIntro} reveals a threshold phenomenon for random RS codes. Specifically, $\cP$ is highly likely to be satisfied when $R > R_{\cP}$ and almost certainly not satisfied when $R < R_{\cP}$. The sharpness of this transition increases rapidly with $q$. Crucially, the threshold matches that of RLCs, implying that an RLC and a random RS code with similar rates are likely to satisfy the same local properties. Indeed, any positive result about local properties of RLCs can immediately be transferred to random RS codes, and vice versa. We illustrate this paradigm by the following corollary.
    
    \begin{corollary}[List-recoverability of random RS codes]\label{cor:RSLR}
		Fix $\rho\in [0,1]$,  $L,\ell,n\in \N$ and $\eps > \frac {c''}n$. Let $k = Rn$ where $R\in [0,1]$ and let $q$ be a prime power satisfying $$q \ge 2^{\frac c\eps}\cdot n\eperiod$$
          
        Let $\cC\subseteq \F_q^n$ be an RLC of rate $R$ and let $\cD\subseteq \F_q^n$ be a random RS code of rate $R':=R-\eps - \frac{c'}n$, where $\eps \ge \frac{c''}n$. Suppose that $$\PR{\cC \textrm{ is }\LR \rho \ell L} \ge \frac 12\eperiod$$ Then, $$\PR{\cD \textrm{ is }\LR\rho \ell L} \ge 1 -  2^{-n}\eperiod$$

        Here, $c$, $c'$ and $c''$ are positive constants that depend polynomially on $\rho$, $\ell$ and $L$.
	\end{corollary}
	\begin{proof}
       Let $\cP$ denote the complement to the property of being $\LR \rho \ell L$ and let $\cF$ be its associated local profile family. By \cref{prop:LDLRareLDL}, $|F| \le \binom n{\rho n}^{L+1}\cdot \ell^{(L+1)n}$.

        Write $R = R_{\cP}+\nu$ for some $\nu \in \R$. If $\nu > 0 $ then, by \cref{thm:RLCThresholdForLCL}
        $$\frac 12 \le \PR{\cC\textrm{ is }\LR \rho \ell L} \le q^{-\nu n+b^2} \ecomma$$
        so $\nu \le \frac{b^2+1}n$. Hence, taking $c' = b^2+1$, we have $R' \le R_{\cP} - \eps$. Take $c$ and $c''$ large enough to satisfy the prerequisites of \cref{thm:MainRSIntro}. The theorem then yields
		\begin{align*}
			\PR{\cD\text{ is }\LR \rho \ell L}  &= \PR{\cD\text{ does not satisfy }\cP}\\&\ge 1- (2^b-1)\cdot \inparen{\frac{(4(L+1))^{4(L+1)}\cdot k}{\eps q}}^{\frac{\eps n}{2(L+1)}}\cdot (2\ell)^{(L+1)n}\eperiod
		\end{align*}
        For $c$ large enough, the right hand side is at least $1-2^{-n}$.        
	\end{proof}
	
	Finally, the following immediate corollary of \cref{thm:RLCThresholdForLDIntro,thm:MainRSIntro}, is very similar to the main theorem of \cite{AGL2023}. Our required lower bound on $q$ is somewhat more demanding than the $q \ge n+k\cdot 2^{10 L /\eps}$ bound stipulated by \cite{AGL2023}.
	\begin{corollary}[List-decodability of random RS codes]\label{cor:RSLD}
		Fix $\rho\in [0,1]$, $L,n\in \N$ and $\eps > 0$ such that $\eps n\ge 2(L+1)(L+2)$. Let $k = Rn$ where $R \le 1-\rho\cdot\inparen{1+\frac 1L} - \eps$ and let $q$ be a prime power satisfying $$q \ge 2^{2(L+1)\cdot\frac{L+1+\eta}\eps}\cdot \frac{k\cdot(4(L+1))^{4(L+1)}}\eps\eperiod$$ Let $\cC\subseteq \F_q^n$ be a random RS code of dimension $k$. Then, $\cC$ is $\ALD \rho L$ with probability at least $1 - 2^{-\eta \cdot n}$.
	\end{corollary}

	\section{LCL Properties of Random Linear Codes}\label{sec:RLCLCL}
	The following is a basic property of an RLC.
	\begin{lemma}[Probability that an RLC contains a set]\label{lem:ProbInRLC}
		Let $\cC\subseteq \F_q^n$ be an RLC of rate $R$ and let $A\in \F_q^{n\times b}$. Then
		$$\PR{A\subseteq \cC} = q^{-(1-R)n\cdot \rank A}$$
	\end{lemma}
	\begin{proof}
		Write $\cC = \ker P$ where $P\in \F_q^{(1-R)n\times n}$ is uniformly random. Now,
		$$\PR{A\subseteq \cC} = \PR{\forall i\in [(1-R)n]~~\mrow Pi\cdot A = 0} = \prod_{i=1}^{(1-R)n}\PR{\mrow Pi\cdot A = 0} = q^{-(1-R)\cdot n\cdot \rank A}\eperiod$$
	\end{proof}
	
	We are interested in the probability that an RLC $\cC$ of rate $R$ satisfies a given $b$-LCL property $\cP$. In particular, we would like to estimate $\PR{\cC \textrm{ contains }\cV}$ for a given linear profile $\cV\in \cL\inparen{\F_q^b}^n$. As a first-order estimate, it makes sense to first compute the expectation
	\begin{equation}\label{eq:M_vExpectation}
		\E{\inabset{A\in \cM_{\cV}^{\distinct}\mid A\subseteq \cC}} = \sum_{A\in \cM_{\cV}^{\distinct}}\PR{A\subseteq \cC} = \sum_{A\in \cM_{\cV}^{\distinct}} q^{-(1-R)n\cdot \rank A}\ecomma
	\end{equation}
	where the last transition is due to \cref{lem:ProbInRLC}. As we shall see, to understand the likelihood of $\cC$ containing $\cV$ we will need a more nuanced expectation argument, in which we classify the matrices in $\cM_\cV$ according to their row span. For each $U\in \cL\inparen{\F_q^b}$ we define 
	$$\cM_{\cV,U} = \inset{A\in \cM_\cV \mid \rspn(A) = U}.$$
	If there exists $A\in \cM_{\cV,U}$ such that $A\subseteq\cC$, we say that \deffont{$\cC$ contains $(\cV,U)$}. 
	Let $$\cL_{\distinct}\inparen{\F_q^b} = \inset{U\in \cL\inparen{\F_q^b}\mid \forall 1\le i < j\le b \:\: \exists x\in U \textrm{ such that }x_i\ne x_j}$$
	and note that a matrix $A\in \cM_{\cV}$ belongs to $\cM_{\cV}^{\distinct}$ if and only if $\rspn(A) \in \cL_{\distinct}\inparen{\F_q^b}$. Hence,
	$$\max_{U\in \cL_{\distinct}\inparen{\F_q^b}}\PR{\exists A\in \cM_{\cV,U},~A\subseteq \cC} \le\PR{\cC \textrm{ contains }\cV} \le \sum_{U\in \cL_{\distinct}\inparen{\F_q^b}} \PR{\exists A\in \cM_{\cV,U},~A\subseteq \cC}\eperiod$$
	Observe that the left-hand and right-hand sides differ by a factor of at most $\inabs{\cL\inparen{\F_q^b}}\le q^{b^2}$. For constant $b$, this is merely polynomial in $q$ and will thus end up being negligible. 
	We therefore turn to estimating the probability that $\cC$ contains $(\cV,U)$ for a given fixed $U\in \cL\inparen{\F_q^b}$. Analogously to \cref{eq:M_vExpectation}, we have
	\begin{equation}\label{eq:M_vuExpectation}
		\E{\inabset{A\in \cM_{\cV,U}\mid A\subseteq \cC}}= \sum_{A\in \cM_{\cV,U}} q^{-(1-R)n\cdot \rank A} = \sum_{A\in \cM_{\cV,U}} q^{-(1-R)n\cdot \dim U} = \inabs{\cM_{\cV,U}}\cdot q^{-(1-R)n\cdot \dim U}\eperiod
	\end{equation}
	To bound the right-hand side we need to bound the term $\inabs{\cM_{\cV,U}}$. To this end, we define 
	$$\cM^*_{\cV,U} = \inset{A\in \cM_\cV \mid \rspn(A) \subseteq U}$$
	and note that $\cM_{\cV,U}\subseteq \cM^*_{\cV,U}$. Notice that even when
	$U\in \cL_{\distinct}\inparen{\F_q^b}$, the set $\cM^*_{\cV,U}$ will also include some matrices with non-distinct columns, in contrast to $\cM_{\cV,U}$. In particular, observe that the all-zero matrix is always an element of $\cM^*_{\cV,U}$.
	
	The cardinality of $\cM^*_{\cV,U}$ is easier to compute than that of $\cM_{\cV,U}$ since the former is a linear subspace of $\F_q^{n\times b}$. Indeed, a matrix $A\in \F_q^{n\times b}$ belongs to $\cM^*_{\cV,U}$ if and only if $\mrow Ai \in \cV_i\cap U$ for all $i\in [n]$. Hence, 
	\begin{equation}\label{eq:dimM*}
		\dim \cM^*_{\cV,U} = \sum_{i=1}^n \dim\inparen{\cV_i\cap U}\eperiod
	\end{equation}
	Defining the potential function \footnote{Although $n$ does not appear explicitly as a parameter to $\pdeg(\cV,U,R)$, it is given implicitly as the length of the profile $\cV$. Thus, \cref{eq:degDef} is indeed a proper definition.}
	\begin{equation}\label{eq:degDef}
		\pdeg\inparen{\cV,U,R} := \sum_{i=1}^n \dim\inparen{\cV_i\cap U} - (1-R)n\cdot \dim U\ecomma
	\end{equation}
	we conclude from \cref{eq:M_vuExpectation} that
	\begin{align*}\E{\inabset{A\in \cM_{\cV,U}\mid A\subseteq \cC}} &= \inabs{\cM_{\cV,U}}\cdot q^{-(1-R)n\cdot \dim U} \le \inabs{\cM^*_{\cV,U}}\cdot q^{-(1-R)n\cdot \dim U} \\&= q^{\dim \cM^*_{\cV,U}}\cdot q^{-(1-R)n\cdot \dim U} \\&= q^{\pdeg\inparen{\cV,U,R}}\eperiod\numberthis \label{eq:degBound}\end{align*}
	
	When $\pdeg(\cV,U,R) \ge \Omega(n)$ we shall say that the potential is \deffont{large}, whereas a potential satisfying $\pdeg(\cV,U,R) \le -\Omega(n)$ will be called \deffont{small}. In our informal discussion we will ignore the middle case in which a potential is neither large nor small. To justify this dichotomy, observe that a potential is ``almost always'' either large or small. Namely, if $-\omega(n)\le \pdeg(\cV,U,R) \le \omega(n)$ then for any fixed $\eps$ we have $\pdeg(\cV,U,R+\eps) = \pdeg(\cV,U,R) + \eps n\cdot \dim U \ge \Omega(n)$ and $\pdeg(\cV,U,R-\eps) = \pdeg(\cV,U,R) - \eps n\cdot \dim U \le -\Omega(n)$.
	
	\cref{eq:degBound} yields an immediate upper bound on the probability that $\cC$ contains $(\cV,U)$. Namely, Markov's bound yields
	\begin{equation}\label{eq:degMarkov}
		\PR{\exists A\in \cM_{\cV,U},~A\subseteq \cC} =  \PR{\inabset{A\in \cM_{\cV,U}\mid A\subseteq \cC} \ge 1} \le \E{\inabset{A\in \cM_{\cV,U}\mid A\subseteq \cC}} \le q^{\pdeg\inparen{\cV,U,R}}\eperiod
	\end{equation}
	Hence, $\pdeg\inparen{\cV,U,R}$ being large is a necessary condition for $\cC$ to be likely to contain $(\cV,U)$. One immediately wonders whether this is also a sufficient condition. The answer turns out to be no. In the following example, $\cC$ is very unlikely to contain $(\cV,U)$ despite $\pdeg\inparen{\cV,U,R}$ being large.
	\
	\begin{example}\label{ex:degMarkovNotTight}
		Let $b=2$ and $U = \F_q^2$. Define $\cV\in \cL\inparen{\F_q^2}^n$ by
		$$
		\cV_i = \begin{cases}
			\F_q^2&\text{if }1\le i\le \frac n 2\\
			\inset{x\in \F_q^2\mid x_1=x_2}&\text{if }\frac n2\le i\le n\eperiod
		\end{cases}
		$$
		Then,
		\begin{align*}
			\pdeg\inparen{\cV,\F_q^2,R} &= \sum_{i=1}^n \dim\inparen{\cV_i\cap \F_q^2} - (1-R)n\cdot \dim \F_q^2 =  \sum_{i=1}^n \dim\inparen{\cV_i} - 2(1-R)n \\ &= \frac n2\cdot 2 + \frac n2\cdot 1 - 2(1-R)n = n\cdot\inparen{2R-\frac 12}\eperiod
		\end{align*}
		In particular, taking $R = \frac 13$ yields $\pdeg\inparen{\cV,\F_q^2,R} = \frac n6 \ge \Omega(n)$, which is large. 
		
		We claim that, in spite of the above, an RLC $\cC$ of rate $\frac 13$ is very unlikely to contain $(\cV,\F_q^2)$. Indeed, suppose that $\cC$ contains a matrix $A\in \cM_{\cV,\F_q^2}$. Let $A'\in \F_q^{n\times 1}$ consist of the difference between the two columns of $A$.  Observe that $A'$ must belong to the set $\cM_{\cV',\F_q^1}$, where $\cV'\in \cL\inparen{\F_q^1}^n$ is given by
		$$
		\cV'_i = \begin{cases}
			\F_q&\text{if }1\le i\le \frac n 2\\
			\{0\}&\text{if }\frac n2\le i\le n\eperiod
		\end{cases}
		$$
		Indeed, since $A$ satisfies $\cV$, the matrix $A'$ must satisfy $\cV'$, and since $\rspn(A) = \F_q^2$, it must hold that $\rspn(A') = \F_q$. Thus, to contain $(\F_q^2,\cV)$, the code $\cC$ must also contain $(\F_q^1,\cV')$. However, \cref{eq:degMarkov} bounds that probability of the latter event by $q^{\pdeg\inparen{\cV',\F_q^1,R}}$, where
		$$\pdeg\inparen{\cV',\F_q^1,R} = \sum_{i=1}^n \dim\inparen{\cV'_i} - (1-R)n = \frac n2 - (1-R)n = n\cdot\inparen{R-\frac 12}\eperiod$$
		Taking $R= \frac 13$ yields
		$$\PR{\cC \textrm{ contains }(\cV,\F_q^2)}\le \PR{\cC \textrm{ contains }(\cV',\F_q^1)} \le q^{\pdeg\inparen{\cV',\F_q^1,\frac13}} = q^{-\frac n6} \le q^{-\Omega(n)}\eperiod$$
	\end{example}
	
	In \cref{ex:degMarkovNotTight}, the expectation bound based on $\pdeg(\cV,U,R)$ is not tight because, informally, the constraints presented by $\cV$ are \emph{skewed} towards a certain part of $\F_q^2$ (in this case, the difference between the two coordinates). Thus, a tighter bound is obtained by considering the potential $\pdeg(\cV',U',R)$ (here $U' = \F_q^1$), which we informally think of as \deffont{implied} by $(\cV,U)$. As we show in the following \deffont{threshold proposition}, considering $\pdeg(\cV,U)$ as well as $\pdeg(\cV',U')$ for all $(\cV',U')$ implied by $(\cV,U)$, yields a necessary and sufficient condition for $\cC$ to contain $(\cV,U)$ with high probability.
	
	\begin{proposition}[RLC Thresholds for linear profiles]\label{prop:RLCThresholdForProfile}
		Let $n\in \N$, and let $q$ be a prime power (which may depend on $n$). Let $\cC\subseteq \F_q^n$ be an RLC of rate $R\in [0,1]$. 
		
		Fix $b\in \N$ and $U\in \cL\inparen{\F_q^b}\setminus \inset{\inset 0}$ and let $\cV\in \inparen{\cL\inparen{\F_q^b}}^n$. Let $$M = \min\inset{\pdeg\inparen{\cV,U,R} - \pdeg\inparen{\cV,W,R}\mid W\in \cL\inparen{\F_q^b}\textrm{ and }W\subsetneq U}\eperiod$$ The following then holds.
		\begin{enumerate}
			\item If $M < 0$ then $\PR{\cC \textrm{ contains }(\cV,U)} \le q^{M}$.
			\item If $M > 0$ then $\PR{\cC \textrm{ contains }(\cV,U)} \ge 1 - q^{-M+b^2}$.
		\end{enumerate}
	\end{proposition}
	Before proving \cref{prop:RLCThresholdForProfile}, we discuss the proposition and its implications. We first note that $\pdeg\inparen{\cV,\inset0,R} = 0 $ for every $\cV$ and $R$. Hence, the condition $M > 0$ in the proposition implies in particular that $\pdeg\inparen{\cV,U,R} = \pdeg\inparen{\cV,U,R} - \pdeg\inparen{\cV,\inset 0,R} > 0$. As we have already seen in \cref{eq:degBound}, large potential is indeed a necessary condition for $\cC$ to be likely to contain $(\cV,U)$.
	
	Let us now revisit \cref{ex:degMarkovNotTight} in light of \cref{prop:RLCThresholdForProfile}. In that example, consider the vector space $W = \inset{x\in \F_q^2\mid x_1=x_2}$. It is not difficult to see that $\pdeg\inparen{\cV,W,R} = Rn$, which is larger than $\pdeg\inparen{\cV,\F_q^2,R} = \inparen{2R-\frac 12}n$ whenever $R < \frac{1}{2} $. \cref{prop:RLCThresholdForProfile} thus implies that an RLC of rate $< \frac{1}{2}$ is unlikely to contain $\inparen{\cV,\F_q^2}$. In \cref{ex:degMarkovNotTight} we reached a similar conclusion by projecting the rows of a matrix $A\in \cM_{\cV,\F_q^2}$ onto their difference. The kernel of this projection $\F_q^2\to \F_q^1$ is exactly $W$. The proof of the first part of \cref{prop:RLCThresholdForProfile} generalizes this method. The second part of the proposition says that this method is essentially tight.
	
	We refer to \cref{prop:RLCThresholdForProfile} as a \deffont{threshold result} since, given $U\in \cL\inparen{\F_q^b}\setminus \inset{\inset 0}$ and $\cV\in \inparen{\cL\inparen{\F_q^b}}^n$, the proposition gives a \deffont{threshold rate} $R_{\cV,U}$ such that, essentially, an RLC of rate below the threshold is very unlikely to contain $(\cV,U)$, while an RLC of rate above the threshold almost surely contains $(\cV,U)$. The threshold rate is given by
	\begin{equation}\label{eq:RVUDef}
		R_{\cV,U} := \max_{\substack{W\in \cL\inparen{\F_q^b}\\ W\subsetneq U}} \inset{1-\frac{\sum_{i=1}^n \inparen{\dim\inparen{\cV_i\cap U}- \dim \inparen{\cV_i\cap W}}}{n\cdot \inparen{\dim U - \dim W}}}\eperiod
	\end{equation}    
	Indeed, one can verify that the term $M$ in \cref{prop:RLCThresholdForProfile} is positive if and only if $R > R_{\cV,U}$. In other words,
        \begin{equation}\label{eq:RVUDefAlt}
            R_{\cV,U} = \min\inset{R\in [0,1]\mid \pdeg(\cV,U,R) \ge \pdeg (\cV,W,R) ~\textrm{ for every linear subspace } W\subseteq U}\eperiod
        \end{equation}
	
	Recall that a code $\cC$ is said to \deffont{contain $\cV$} if and only if it contains $\inparen{\cV,U}$ for some $U\in \cL_{\distinct}\inparen{\F_q^{b}}$. We thus define 
	\begin{equation}\label{eq:RVDef}
		R_{\cV} := \min_{U\in \cL_{\distinct}\inparen{\F_q^{b}}} \inset{R_{\cV,U}} = \min_{U\in \cL_{\distinct}\inparen{\F_q^{b}}} \max_{\substack{W\in \cL\inparen{\F_q^b}\\ W\subsetneq U}} \inset{1-\frac{\sum_{i=1}^n \inparen{\dim\inparen{\cV_i\cap U}- \dim \inparen{\cV_i\cap W}}}{n\cdot \inparen{\dim U - \dim W}}}\eperiod
	\end{equation}
	
	Let $\cP$ be a $b$-LCL property whose associated profile family is $\cF$. Recall that $\cC$ \deffont{satisfies} $\cP$, if $\cC$ contains some $\cV\in \cF$. We thus define
	\begin{equation}\label{eq:R_PDef}
		R_{\cP} := \min_{\cV\in\cF} \inset{R_\cV} = \min_{\substack{\cV\in\cF\\U\in \cL_{\distinct}\inparen{\F_q^{b}}}} \inset{R_{\cV,U}}\eperiod
	\end{equation}
	We now state and prove \cref{thm:RLCThresholdForLCL}---a more detailed version of \cref{thm:RLCThresholdForLCLIntro}. \cref{thm:RLCThresholdForLCL} states that \cref{eq:R_PDef} correctly characterizes the \deffont{threshold rate} for LCL properties of RLCs.
	\begin{theorem}[More detailed version of \cref{thm:RLCThresholdForLCLIntro}]\label{thm:RLCThresholdForLCL}
		Let $\cP$ be a $b$-LCL property of codes in $\F_q^n$ and let $\cF\subseteq \cL\inparen{\F_q^b}^n$ be the corresponding family of profiles. Let $\cC\subseteq \F_q^n$ be an RLC of rate R. Define $R_\cP$ as in \cref{eq:R_PDef}. The following now holds
		\begin{enumerate}
			\item If $R \ge R_{\cP}+\eps$ then $\PR{\cC\textrm{ satisfies }\cP} \ge 1 - q^{-\eps n+b^2}$.
			\item If $R \le R_{\cP}-\eps$ then $\PR{\cC\textrm{ satisfies }\cP} \le \inabs{\cF}\cdot q^{-\eps n + b^2}$.
			\item In particular, if $R \le R_{\cP}-\eps$ and $q \ge 2^{\frac{2\log_2 |\cF|}{\eps n}}$ then $\PR{\cC\textrm{ satisfies }\cP} \le q^{-\frac{\eps n}2 + b^2}$.
		\end{enumerate}
	\end{theorem}
	\begin{proof}
		Suppose first that $R\ge R_\cP+\eps$.  Let $\cV\in \cF$ and $U\in \cL_{\distinct}\inparen{\F_q^b}$ such that $R_\cP = R_{\cV,U}$, so $R\ge R_{\cV,U}+\eps$. Let $M$ be as in \cref{prop:RLCThresholdForProfile}, namely
		$$M = \min\inset{\pdeg\inparen{\cV,U,R} - \pdeg\inparen{\cV,W,R}\mid W\in \cL\inparen{\F_q^b}\textrm{ and }W\subsetneq U}\eperiod$$
		Fix a subspace $W\subsetneq U$ and write $d = \dim U - \dim W$ and $g = \sum_{i=1}^n \inparen{\dim\inparen{\cV_i\cap U}- \dim\inparen{\cV_i\cap W}}$. By \cref{eq:RVUDef}, $R_{\cV,U} \ge 1- \frac{g}{nd}$. Thus,
		\begin{align*}
			M:=\pdeg(\cV,U,R)- \pdeg(\cV,W,R) = g - (1-R)n\cdot d\ge g - (1-R_{\cV,U}-\eps)n\cdot d
			\ge \eps n d\ge \eps n
		\end{align*}
		
		Therefore, by \cref{prop:RLCThresholdForProfile},
		$$\PR{\cC\textrm{ satisfies }\cP}\ge \PR{\cC\textrm{ contains }(\cV,U)} \ge 1-q^{-M+b^2}\ge 1-q^{-\eps n+b^2}\eperiod$$
		
		We turn to proving the second claim. Let $R \le R_{\cP}-\eps$. Let $\cV\in \cF$ and $U\in \cL_{\distinct}\inparen{\F_q^b}$. Let $W\subsetneq U$ be a linear space such that 
		$R_{\cV,U} = 1-\frac{g}{nd}$ where $d = \dim U - \dim W$ and $g = \sum_{i=1}^n \inparen{\dim\inparen{\cV_i\cap U}- \dim\inparen{\cV_i\cap W}}$. Then, $R\le R_{\cP}-\eps \le R_{\cV,U}-\eps = 1-\frac{g}{nd}-\eps$.
		Hence, 
		$$M_{\cV,U} \leq \pdeg(\cV,U,R)-\pdeg(\cV,W,R) = g-nd(1-R) \le -\eps nd\ecomma$$
		so 
		$$\PR{\cC \textrm{ contains }(\cV,U)}\le q^{-\eps nd}\le q^{-\eps n}$$
		by \cref{prop:RLCThresholdForProfile}.
		Now,
		$$
		\PR{\cC \textrm{ satisfies }\cP} \le \sum_{\cV\in \F}\sum_{U\in \cL_{\distinct}\inparen{\F_q^b}} \PR{\cC \textrm{ contains }(\cV,U)} \le |\cF|\cdot \inabs{\cL_{\distinct}\inparen{\F_q^b}}\cdot q^{-\eps n} = |\cF|\cdot q^{-\eps n + b^2}\eperiod
		$$
		
		Finally, the third claim follows immediately from the second claim.            
	\end{proof}
	
	We turn to proving \cref{prop:RLCThresholdForProfile}.

	\begin{proof}[Proof of \cref{prop:RLCThresholdForProfile}]
		We begin with the first statement. Suppose that $M < 0$ and let $W\in \cL\inparen{\F_q^b}$ such that $\pdeg(V,U,R) - \pdeg(V,W,R) = M$. Write $d = \dim U- \dim W$ and let $\varphi:U\to \F_q^d$ be a linear map such that $\ker \varphi = W$. Let $B\in \F_q^{b\times d}$ such that $xB = \varphi(x)$ for all $x\in U$.
		
		Since $\cC$ is a linear code, to contain a matrix $A\in \F_q^{n\times b}$ the code must also contain the matrix $AB$, whose columns are merely linear combinations of the columns of $A$. Hence, in order to contain a matrix from $\cM_{\cV,U}$, the code $\cC$ must also contain some matrix from the set $\cM_{\cV,U}B:=\inset{AB\mid A\in \cM_{\cV,U}}$. Observe that every matrix in the latter set is of rank $d$. Indeed, if $A\in \cM_{\cV,U}$ then $$\rspn (AB) = \varphi\inparen{\rspn(A)} = \varphi(U) = \F_q^d\eperiod$$ Therefore, by \cref{lem:ProbInRLC} 
		\begin{align*}
			\PR{\cC \textrm{ contains }(V,U)} &\le \PR{\exists D\in \cM_{\cV,U}B,~D\subseteq \cC}\le \sum_{D\in \cM_{\cV,U}B}{\PR{D\subseteq \cC}} \\&= \sum_{D\in \cM_{\cV,U}B}q^{-(1-R)n\cdot \rank D} = \inabs{\cM_{\cV,U}B}\cdot q^{-(1-R)nd}
		\end{align*}
		
		To bound $\inabs{\cM_{\cV,U}B}$, we consider the linear space $\cM_{\cV,U}^*B:=\inset{AB\mid A\in \cM^*_{\cV,U}}$, which clearly contains $\cM_{\cV,U}B$. Thus,
		\begin{align*}
			\log_q\inabs{\cM_{\cV,U}B} &\le \log_q\inabs{\cM_{\cV,U}^*B} = \dim \cM_{\cV,U}^*B  = \sum_{i=1}^n \dim\inparen{\varphi\inparen{\cV_i\cap U}} \\ 
			&= \sum_{i=1}^n \inparen{\dim\inparen{\cV_i\cap U}-\dim\inparen{\ker \varphi\cap \cV_i\cap U}} \\ 
			&= \sum_{i=1}^n \inparen{\dim\inparen{\cV_i\cap U}-\dim\inparen{\cV_i\cap W}}\eperiod
		\end{align*}
		Therefore,
		\begin{align*}
			\PR{\cC \textrm{ contains }(V,U)} \le \inabs{\cM_{\cV,U}B}\cdot q^{-(1-R)nd} &\le q^{\sum_{i=1}^n \inparen{\dim\inparen{\cV_i\cap U}-\dim\inparen{\cV_i\cap W}}-(1-R)nd} \\ &= q^{\sum_{i=1}^n \inparen{\dim\inparen{\cV_i\cap U}-\dim\inparen{\cV_i\cap W}}-(1-R)n\inparen{\dim U - \dim W}} \\
			&= q^{\pdeg\inparen{\cV,U,R} - \pdeg\inparen{\cV,W,R}} = q^M\eperiod
		\end{align*}
		
		We turn to proving the second claim. Suppose that \begin{equation}\label{eq:DegDifferenceLowerBound}
			\pdeg(\cV,U,R) - \pdeg(\cV,W,R) \ge M
		\end{equation} 
		for every proper linear subspace $W\subsetneq U$, for some positive $M$. 
		
		For any $b\in \N$, denote
		$\cC^b := \inset{A\in \F_q^{n\times b} \mid A \subseteq \cC}$. Given $W\in \cL\inparen{\F_q^b}$, let
		\[
		\F_q^{n\times W} = \inset{A\in \F_q^{n\times b}\mid \rspn(A)\subseteq W} \eperiod
		\]
		Note that $\dim (\cC^b) = Rnb$ and $\dim (\F_q^{n\times W}) = n\cdot \dim W$. We claim that 
		\begin{equation}\label{eq:dimC_W}
			\dim \inparen{\cC^b\cap \F_q^{n\times W}} = Rn\cdot \dim W\eperiod
		\end{equation}
		Indeed, suppose without loss of generality that the first $\dim W$ coordinates are an information set for $W$. In other words, the projection of a vector $x\in W$ onto its first $\dim W$ coordinates is a bijection from $W$ onto $\F_q^{\dim W}$. Now, consider the linear transformation $\pi:\cC^b\cap \F_q^{n\times W}\to \cC^{\dim W}$ which maps a matrix to the submatrix consisting of its first $\dim W$ columns. It follows readily from the linearity of $\cC$ that $\pi$ is bijective, which implies \cref{eq:dimC_W}. \cref{eq:dimM*,eq:dimC_W} now yield
		\begin{align*}
			\dim\inparen{\cM^*_{\cV,U}\cap \cC^b} &= \dim\inparen{\cM^*_{\cV,U}\cap \cC^b\cap \F_q^{n\times U}} \\&= \dim \cM^*_{\cV,U} + \dim \inparen{\cC^b\cap \F_q^{n\times U}} - \dim \inparen{\cM^*_{\cV,U}+\inparen{\cC^b\cap\F_q^{n\times U}}} 
			\\&\ge \dim \cM^*_{\cV,U} + \dim \inparen{\cC^b\cap \F_q^{n\times U}} - \dim\inparen{\F_q^{n\times U}} \\
			&=\sum_{i=1}^n\dim\inparen{\cV_i\cap U} + Rn\cdot \dim U - n\cdot \dim U \\
			&=\sum_{i=1}^n\dim\inparen{\cV_i\cap U} - (1-R)n\cdot \dim U = \pdeg(\cV,U,R)\eperiod \numberthis\label{eq:dimM*DegBound}
		\end{align*}
		
		Let $A\in \F_q^{n\times b}$ such that $\rspn(A) \subseteq U$. Note that $\rspn(A) = U$ if and only if $\rspn(A) \ne W$ for every linear space $W\subsetneq U$. Thus,         
		$$\inabs{\cM_{\cV,U}\cap \cC^b} = \inabs{\inparen{\cM^*_{\cV,U}\cap \cC^b}\setminus \bigcup_{W\subsetneq U} \inparen{\cM_{\cV,W}\cap \cC^b}}\ge q^{\pdeg(\cV,U,R)}-\sum_{W\subsetneq U}{\inabs{\cM_{\cV,W}\cap \cC^b}\eperiod}$$
		Therefore, by Markov's inequality,
		\begin{align*}
			\PR{\cC \textrm{ does not contain }(V,U)} &= \PR{\cM_{\cV,U}\cap \cC^b = \emptyset} \\&\le \PR{\sum_{W}{\inabs{\cM^*_{\cV,W}\cap \cC^b}\ge q^{\pdeg(\cV,U,R)}}} \\
			&\le \sum_W\E{\inabs{\cM^*_{\cV,W}\cap \cC^b}}\cdot q^{-\pdeg\inparen{\cV,U,R}} \\
			&\le \sum_W q^{\pdeg\inparen{\cV,W,R}}\cdot q^{-\pdeg\inparen{\cV,U,R}} &\textrm{by \cref{eq:degBound}} \\
			&\le q^{b^2}\cdot q^{\pdeg\inparen{\cV,W,R}}\cdot q^{-\pdeg\inparen{\cV,U,R}} & \textrm{since }\inabs{\cL\inparen{\F_q^b}}\le q^{b^2}\\
			&\le q^{b^2-M} & \textrm{by \cref{eq:DegDifferenceLowerBound}}\eperiod
		\end{align*}
	\end{proof}
	
	We end this section with a useful fact. It turns out that threshold rate $R_\cV$, defined in \cref{eq:RVDef}, has another convenient characterization. 
	Concretely, we have the following lemma.
	\begin{lemma}\label{lem:RVAlt}
		Let $n,b\in \N$, $q$ a prime power and $\cV\in \inparen{\cL\inparen{\F_q^b}}^n$. For $R\in [0,1]$, denote
		$$\argmax \inset{\pdeg(\cV,*,R)} \coloneqq \inset{U\in \cL\inparen{\F_q^b}\mid \pdeg(\cV,U,R) = \max_{W\in \cL\inparen{\F_q^b}}\inset{\pdeg(\cV,W,R)}}\eperiod$$            
		The following then holds:
		\begin{enumerate}
			\item For all $R\in [0,R_\cV)$ we have $\argmax \inset{\pdeg(\cV,*,R)} \subseteq \cL\inparen{\F_q^b} \setminus\cL_\distinct\inparen{\F_q^b}$.
			\item For all $R\in (R_\cV,1]$ we have $\argmax \inset{\pdeg(\cV,*,R)} \subseteq \cL_\distinct\inparen{\F_q^b}$.
			\item $\argmax\inset{\pdeg(\cV,*,R_{\cV})}$ contains an element of $\cL\inparen{\F_q^b} \setminus\cL_\distinct\inparen{\F_q^b}$ and an element of $\cL_\distinct\inparen{\F_q^b}$.
		\end{enumerate}
	\end{lemma}
	The following is an immediate corollary obtained by taking the contrapositives of the two statements in \cref{lem:RVAlt}.
	\begin{corollary}\label{cor:RVAltContrapos}
		Let $n,b\in \N$, $q$ a prime power and $\cV\in \inparen{\cL\inparen{\F_q^b}}^n$. Let $R\in [0,1]$. The following holds.
		\begin{enumerate}
			\item If $\argmax \inset{\pdeg(\cV,*,R)}$ contains an element of $ \cL_{\distinct}\inparen{\F_q^b}$ then $R\ge R_\cV$.
			\item If $\argmax \inset{\pdeg(\cV,*,R)}$ contains an element of $ \cL\inparen{\F_q^b}\setminus\cL_{\distinct}\inparen{\F_q^b}$ then $R\le R_\cV$.
		\end{enumerate}        
	\end{corollary}
	
	Our proof of \cref{lem:RVAlt} relies on the following lemma, which we prove immediately after \cref{lem:RVAlt}.
	\begin{lemma}\label{lem:degOfCap}
		Let $n,b\in \N$, $q$ a prime power, $\cV\in \inparen{\cL\inparen{\F_q^b}}^n$. Fix $R\in [0,1]$ and let $U,W\in \cL\inparen{\F_q^b}$. Then, $$\pdeg(\cV,U,R)+\pdeg(\cV,W,R) \le \pdeg(\cV,U\cap W, R)+ \pdeg(\cV,U+W,R)\eperiod$$
	\end{lemma}
	\begin{proof}[Proof of \cref{lem:RVAlt}]
		The third statement follows from the first two since $\pdeg(\cV,U,R)$ is continuous in $R$.
		
		We will prove the contrapositives of the first two statements. For the first statement, suppose there exists some $U\in \argmax\inset{\pdeg(V,*,R)}\cap \cL_\distinct\inparen{\F_q^b}$. We need to show that $R\ge R_\cV$. Indeed, by \cref{eq:RVDef},  
		\begin{align*}
			R_{\cV} &\le R_{\cV,U} = \max_{\substack{W\in \cL\inparen{\F_q^b}\\ W\subsetneq U}} \inset{1-\frac{\sum_{i=1}^n \inparen{\dim\inparen{\cV_i\cap U}- \dim \inparen{\cV_i\cap W}}}{n\cdot \inparen{\dim U - \dim W}}}
		\end{align*}
		Fix $W\subsetneq U$. To prove that $R\ge R_\cV$ it suffices to show that
		$$R \ge 1-\frac{\sum_{i=1}^n \inparen{\dim\inparen{\cV_i\cap U}- \dim \inparen{\cV_i\cap W}}}{n\cdot \inparen{\dim U - \dim W}}\eperiod$$
		The above is equivalent to $\pdeg(\cV,U,R)\ge \pdeg(\cV,W,R)$, which follows from our assumption that $U\in \argmax\inset{\pdeg(V,*,R)}\cap \cL_\distinct\inparen{\F_q^b}$
		
		We turn to the second statement. We now assume that there is some $W\in \argmax\inset{\pdeg(V,*,R)}\cap \inparen{\cL\inparen{\F_q^b}\setminus\cL_\distinct\inparen{\F_q^b}}$ and need to prove that $R \le R_\cV$. Fix $U\in \cL_\distinct\inparen{\F_q^b}$. By \cref{eq:RVDef}, it suffices to prove
		that $R \le R_{\cV,U}$. By \cref{eq:RVUDef}, the latter would follow if we show that
		\begin{equation}\label{eq:RVAltNeedtoshow}
			R \le 1-\frac{\sum_{i=1}^n \inparen{\dim\inparen{\cV_i\cap U}- \dim \inparen{\cV_i\cap W'}}}{n\cdot \inparen{\dim U - \dim W'}}
		\end{equation}
		for some $W' \subsetneq U$. We take $W' = W\cap U$. Note that $W'$ is indeed strictly contained in $U$ since $W\notin \cL_{\distinct}\inparen{\F_q^b}$ and thus so does any subspace of $W$, including $W'$. By assumption, $U\in \cL_{\distinct}\inparen{\F_q^b}$, so $W'\ne U$ and thus the containment is indeed strict. Now, \cref{eq:RVAltNeedtoshow} is equivalent to the statement
		$$\pdeg(\cV,U,R) \le \pdeg(\cV,U\cap W, R)\eperiod$$
		By \cref{lem:degOfCap}, 
		$$\pdeg(\cV,U,R) \le \pdeg(\cV,U\cap W, R) + \pdeg(\cV,U+ W, R) - \pdeg(\cV,W, R)\le \pdeg(\cV,U\cap W, R)\ecomma$$
		where the second inequality is due to the maximality of $\pdeg(\cV,W,R)$.
	\end{proof}
	
	\begin{proof}[Proof of \cref{lem:degOfCap}]
		Fix $Z\in \cL\inparen{\F_q^b}$. Now,
		
		\begin{align*}
			\dim\inparen{U\cap W\cap Z} + \dim\inparen{(U+W)\cap Z}&\ge 
			\dim\inparen{U\cap W\cap Z} + \dim\inparen{(U\cap Z) + (W\cap Z)} \\
			&=\dim(U\cap Z) + \dim(W\cap Z)
		\end{align*}
		Therefore,
		\begin{align*}
			&\phantom{{}={}}\pdeg(\cV,U\cap W,R) + \pdeg(\cV,U+W,R) \\&= \sum_{i=1}^n \inparen{\dim \inparen{U\cap W\cap \cV_i} + \dim\inparen{(U+W)\cap \cV_i} - (1-R)\cdot \dim(U\cap W) - (1-R)\cdot \dim(U+W)} \\
			&\ge \sum_{i=1}^n \inparen{\dim\inparen{U\cap \cV_i} + \dim\inparen{W\cap \cV_i} - (1-R)\cdot \dim(U\cap W) - (1-R)\cdot \dim(U+W)} \\
			&=\sum_{i=1}^n \inparen{\dim\inparen{U\cap \cV_i} + \dim\inparen{W\cap \cV_i} - (1-R)\cdot \dim(U) - (1-R)\cdot \dim(W)} \\
			&= \pdeg(\cV,U,R) + \pdeg(\cV,W,R)\eperiod
		\end{align*}

	\end{proof}
	
	\section{List-Decodability and List-Recoverability of RLCs}\label{sec:RLCLD} 
        In this section we use the characterization of the threshold rate, given in \cref{eq:R_PDef}, to compute the threshold for list-decodability and average-weight list-decodability of RLCs. We then prove an upper bound on the threshold for list-recoverability and average-weight list-recoverability of RLCs. 

    \subsection{Threshold for List-Decodability}
    In this section we prove the following result.
    	\RLCLD*
        \begin{proof}~
        \subsubsection*{Notation for This Proof:}
		We first set up some notation.        
		Given a set $Z\subseteq [L+1]$, let $$D_Z := \inset{x\in \F_q^{L+1}\mid \supp(x)\subseteq Z},$$
        and let 
        $$E_Z:= \inset{x \in \F_q^{L+1} \mid \forall ~i,j\in Z, i=j}\eperiod$$
		Let $$E := E_{[L+1]} = \inset{x\in \F_q^{L+1}\mid \forall i,j~~x_i = x_j}\eperiod$$        
		
		Let $\cF,\cF'\subseteq\cL\inparen{\F_q^{L+1}}^n$ denote the local profiles families associated with $\cP$ and $\cP'$, respectively. Observe that every profile in $\cV\in \cF$ is of the form $\cV_i = E_{Z_i}$ ($i\in [n]$), where:
		\begin{itemize}
			\item $Z_1,\dots,Z_n$ are non-empty subsets of $[L+1]$. 
			\item For each $j\in [L+1]$, it holds that $\inabset{i\in [n]\mid j\in Z_i} \ge (1-\rho)n$.
		\end{itemize}
		Profiles in $\cF'$ are of similar form, except that the second condition is relaxed to $\sum_{i=1}^n |Z_i| \ge (L+1)\cdot(1-\rho)\cdot n$.

		Clearly, $R_{\cP}$ and $R_{\cP'}$ are both monotone-decreasing in $\rho$. Hence, it suffices to prove the theorem under the assumption that $\rho \le 1-\frac{1}{L+1}$, and so, $1-\rho\inparen{1+\frac 1L} \ge 0$. Also, note that $R_{\cP'} \le R_\cP$. Thus, to prove the theorem it is enough to  separately prove the inequalities $R_{\cP'} \ge 1-\rho\inparen{1+\frac 1L}$ and $R_\cP \le 1-\rho\inparen{1+\frac 1L}$.

        \subsubsection*{Proof of $R_{\cP'} \ge 1-\rho\inparen{1+\frac 1L}$:}
		Let $R =  1-\rho\inparen{1+\frac 1L}$. We claim that for every $\cV\in\cF'$ and $U\in \cL_\distinct\inparen{\F_q^{L+1}}$ there exists $W\in \cL\inparen{\F_q^{L+1}}\setminus\cL_\distinct\inparen{\F_q^{L+1}}$ such that \begin{equation}\label{eq:degULessThandegW}
			\pdeg\inparen{\cV,U,R} \le \pdeg\inparen{\cV,W,R}\eperiod
		\end{equation}
		This would imply that $\argmax \inset{\pdeg(\cV,*,R)}\cap \inparen{\cL\inparen{\F_q^{L+1}}\setminus\cL_\distinct\inparen{\F_q^{L+1}}}$ is nonempty, whence $R_{\cV} \ge R$ by virtue of \cref{cor:RVAltContrapos}, and therefore $R_{\cP'} \ge R$ by \cref{eq:R_PDef}.
		
		\sloppy
		To prove \cref{eq:degULessThandegW} we consider two cases. First, if $U+E = \F_q^{L+1}$, we take $W = E$. Observe that $E\in \cL\inparen{\F_q^{L+1}}\setminus\cL_\distinct\inparen{\F_q^{L+1}}$ and 
		$$\pdeg(\cV,W,R) = \sum_{i=1}^n\dim(\cV_i\cap E)-(1-R)\dim E\cdot n = \sum_{i=1}^n 1 - (1-R)\cdot n = R\cdot n\eperiod$$
		On the other hand,
		\begin{align*}
			\pdeg\inparen{\cV,U,R} &= \sum_{i=1}^n \dim\inparen{\cV_i\cap U} - (1-R)\cdot \dim U\cdot n = \sum_{i=1}^n \dim(U\cap E_{Z_i}) - (1-R)\cdot\dim U\cdot n \\
			&=
			\sum_{i=1}^n \inparen{\dim U + \dim E_{Z_i} - \dim \inparen{U+E_{Z_i}}} - (1-R)\cdot\dim U\cdot n
			\\
			&= \sum_{i=1}^n \inparen{\dim U + \dim E_{Z_i} - (L+1)} - (1-R)\cdot\dim U\cdot n\\
			&= \sum_{i=1}^n \inparen{\dim U + (L+2-|{Z_i}|) -(L+1)} - (1-R)\cdot\dim U\cdot n \\
                &= \sum_{i=1}^n \inparen{\dim U + 1 - |Z_i|} - (1-R)\cdot\dim U\cdot n \\
			&\le \inparen{\dim U - (1-\rho)\cdot(L+1)+1}\cdot n -(1-R)\cdot \dim U \cdot n \\
			&= \inparen{\dim U - L}\cdot Rn \le Rn = \pdeg(\cV,W,R)\ecomma
		\end{align*}
		which yields \cref{eq:degULessThandegW}. 
		
		\sloppy
		For the second case, suppose that $K:=U+E\subsetneq \F_q^{L+1}$. For a set $J\subseteq [L+1]$, denote $f(J) := \dim\inparen{K\cap D_J}$.
		We claim that there exists a set $T\subseteq [L+1]$ with the following properties:
		\begin{enumerate}
			\item \label[condition]{enum:TLargerthan2} $|T| \ge 2$
			\item \label[condition]{enum:f(T)=1} $f(T) = 1$
			\item \label[condition]{enum:f(T')=0} Every nonempty $T'\subsetneq T$ has $f(T')=0$.
		\end{enumerate}
		We turn to construct such a set. Let $I = \inset{i\in [L+1]\mid e_i\notin K}$. Observe that $I\ne \emptyset$ by our assumption that $K\ne \F_q^{L+1}$. Note that $f(\{i\}) = 0$ for all $i\in I$. 

        On the other hand, we claim that $\1_I\in K\cap D_I$, so $f(I)\ge 1$. Indeed, $\1\in E$, so it is also in $K=U+E$. By definition of $I$ we have $e_i\in K$ for all $i\in [L+1]\setminus I$. Thus, $\1_I = \1 - \sum_{j\in [L+1]\setminus I}e_j\in K\cap D_I$.
		
		Hence, one can take $T\subseteq I$ to be an inclusion-minimal nonempty set with $f(T) \ge 1$. We show that $T$ satisfies \cref{enum:TLargerthan2,enum:f(T)=1,enum:f(T')=0} above. \cref{enum:f(T')=0} follows immediately from the minimality of $T$. \cref{enum:TLargerthan2} holds since $f(\inset i) = 0$ for all $i\in I$, implying that $|T| > 1$. Finally, let some $z\in T$ and observe that $f(T) \le f(T\setminus \{z\}) + 1 = 1$, where the equality is due to \cref{enum:f(T')=0}. Since $f(T)\ge 1$ by the definition of $T$, \cref{enum:f(T)=1} follows.

        Take $W=(K+D_T)\cap E_T$. Note that $W\in \cL\inparen{\F_q^{L+1}}\setminus \cL_\distinct\inparen{\F_q^{L+1}}$. We turn to prove that $\pdeg(\cV,U,R)\le \pdeg(\cV,W,R)$.
        
        For $t\in T$, define the linear map $\varphi_t: K \to W$ by
        $$\varphi_t(x)_j =\begin{cases} x_j, & j\notin T\\ x_t, &j\in T \end{cases}.$$
        We claim that $\varphi_t$ is a bijection.
        First we establish that the image is indeed contained in $W$. It is evident that $\varphi_t(x)\in E_T$ and $\varphi_t(x)-x\in D_T$. Therefore, $\varphi_t(x)=x+ (\varphi_t(x)-x)\in K+D_T$ and thus $\varphi_t(x)\in (K+D_T)\cap E_T =W$.
        
        To prove injectivity of $\varphi_t$, assume there is a non-zero $x\in K$ such that $\varphi_t(x)=0$. Let $T' = T\setminus \{t\}$. The definition of $\varphi$ implies that $x_i = 0$ for all $i\in ([L+1]\setminus T)\cup \{t\}$, so $x\in D_{T'}$. By \cref{enum:f(T')=0}, $K\cap D_{T'} = \{0\}$, so $x=0$.

        Finally, to prove surjectivity of $\varphi_t$, let $w\in W$. Since $W=(K+D_T)\cap E_{T}$, there are $x\in K$ and $y\in D_T$ such that $x+y=w$. By \cref{enum:f(T)=1,enum:f(T')=0}, there is a vector $z\in K$ whose support is exactly $T$. Evidently, $w= \varphi_t (x+\alpha z)$ for some $\alpha \in \mathbb F_q$.

        We now claim that  $\dim(K\cap E_{Z})\leq \dim(W\cap E_{Z})$ for every $Z\subseteq [L+1]$. If $Z\cap T=\emptyset$ choose an arbitrary $t\in T$ and if $Z\cap T\neq \emptyset$ choose some $t\in Z\cap T$. The claim follows since, in both cases, $\varphi_t$ injectively maps $K\cap E_{Z}$ to $W\cap E_{Z}$.

        Note that $\dim W = \dim K$ since there is a bijection between $W$ and $K$. This and the above inequality imply
        \begin{align*}
            \pdeg(\cV,W,R) &= \sum_{i=1}^n \dim(\cV_i\cap W) - (1-R)\cdot n \cdot \dim W \\
                &= \sum_{i=1}^n \dim(E_{{Z_i}}\cap W) - (1-R)\cdot n \cdot \dim K \\
                &\geq \sum_{i=1}^n \dim(E_{{Z_i}}\cap K) - (1-R)\cdot n \cdot \dim K = \pdeg(\cV,K,R) \eperiod
        \end{align*}
        Lastly, denote $d:= \dim K -\dim U$, and observe that
        \begin{align*}
            \pdeg(\cV,K,R) - \pdeg(\cV,U,R)&= \sum_{i=1}^n \inparen{\dim(\cV_i\cap K) -\dim(\cV_i \cap U)} -(1-R)\cdot n \cdot d\\
            &= \sum_{i=1}^n d -(1-R)\cdot n \cdot d = Rnd \geq 0\ecomma
        \end{align*}
        and thus $\pdeg(\cV,W,R) \geq \pdeg(\cV,U,R). $

            \subsubsection*{Proof of $R_{\cP} \le 1-\rho\inparen{1+\frac 1L}$:} Let $R = 1-\rho\inparen{1+\frac 1L}$. We will construct a linear profile $\cV \in \cF$ such that \begin{equation}\label{eq:RLCLDNegativeBoundNTS}
			\pdeg\inparen{\cV, U, R} \leq \pdeg\inparen{\cV, \F_q^{L+1}, R}
		\end{equation} for every subspace $U \subseteq \F_q^{L+1}$. This would imply that $\argmax \inset{\pdeg(\cV,*,R)}\cap \cL_\distinct\inparen{\F_q^{L+1}}$ is nonempty, whence $R_{\cV} \le R$ by virtue of \cref{cor:RVAltContrapos}. \cref{eq:R_PDef} then yields $$R_\cP \le R_\cV \le R = 1-\rho\inparen{1+\frac 1L}\eperiod$$
		
		We now describe the aforementioned linear profile $\cV$. Let $s := (1-\rho)(L+1)$. Note that $0\le \rho\le \frac{L}{L+1}$. Write $t = \binom{L+1}{s}$. Recall that $n$ is divisible by $t$. Let the sequence $Z_1,\dots, Z_n$ consist of $\frac nt$ instances of each subset of $[L+1]$ of size $s$. Let $\cV_i = E_{Z_i}$ for all $i\in [n]$. Observe that $\cV\in \cF$.  We turn to proving \cref{eq:RLCLDNegativeBoundNTS}.
		
		First, observe that $$\dim \cV_i = \dim E_{Z_i} = L+2 - |Z_i| = (L+1)\cdot \rho + 1\ecomma$$ so 
		\begin{align*}
			\pdeg\inparen{\cV, \F_q^{L+1},R} = \sum_{i=1}^n \dim \cV_i - n(L+1)(1-R) &= n\cdot\inparen{(L+1)\cdot \rho + 1 - (L+1)(1-R)} \\&= n\cdot\inparen{\frac{(L+1)(1-\rho)-1}L}\eperiod
		\end{align*} 
		
		Now, let $U\subsetneq \F_q^{L+1}$. We need to show that \cref{eq:RLCLDNegativeBoundNTS} holds with regard to $U$.
		Denote $d \coloneqq \dim U$. We provide an alternative formulation of $\pdeg \inparen{\cV, U, R}$ in terms of linear maps. Let $\phi: \F_q^{L+1} \rightarrow \F_q^{L+1-d}$ be a full rank linear map such that $\ker \phi = U$.  We can write
		\begin{align*}
			\pdeg \inparen{\cV, U, R} &= \sum_i \dim \inparen{\cV_i \cap U} - \inparen{1-R}nd \\
			&= \sum_i \inparen{ \dim \cV_i - \dim \inparen{\phi(\cV_i)}} - \inparen{1-R}nd \\
			&= \sum_i \dim \cV_i - \sum_i \dim \inparen{\phi(\cV_i)} - \inparen{1-R}nd \\ 
			&= \inparen{(L+1)\cdot \rho + 1}\cdot n - \sum_i \dim\inparen{\phi(\cV_i)} - \frac{L - (1-\rho)(L+1)+1}L\cdot nd\eperiod \numberthis\label{eq:degVUR}
		\end{align*}
		
		We turn to bound the term $\sum_i \dim\inparen{\phi(\cV_i)}$.
		Denote $L' \coloneqq L+1-d$. Let $M \in \F_q^{L' \times (L+1)}$ be a matrix representing the linear map $\phi$ in the standard basis. Because $\phi$ is full rank, $M$ has rank exactly $L'$. Therefore, there exists a set of linearly independent columns from $M$ of size exactly $L'$. Denote the coordinates of such a set of columns by $I\subseteq [L+1]$ (so $|I| = L'$). Let $\pi:\F_q^{L+1}\to \F_q^{L'}$ denote the projection map into the coordinate set $I$. Let $M'\in \F_q^{L'\times L'}$ denote the restriction of the matrix $M$ to the columns indicated by $I$, and let $\phi':\F_q^{L'}\to \F_q^{L'}$ be the bijective linear map represented by $M'$. Observe that for any linear space $V\subseteq \F_q^{L+1}$ we have
		$$\dim \phi(V) \ge \dim \phi'(\pi(V)) = \dim \pi(V)\eperiod$$
		
		In particular, if $V = E_Z$ for some $Z\subseteq [L+1]$ then
		$$\dim \phi(V) \ge \dim \pi(V) = L' + 1 - \max\inset{|I\cap Z|,1}\eperiod$$
		Thus,
		$$
		\sum_{i}\dim \phi(\cV_i) = \sum_i \dim \phi(E_{Z_i}) \ge n\cdot \inparen{L'+1 - \Eover{Z}{\max\inset{|I\cap Z|,1}}}\ecomma
		$$
		where the expectation is over a uniformly random set $Z\subseteq [L+1]$ of size $s$.  Since 
		\begin{align*}
			\Eover{Z}{\max\inset{|I\cap Z|,1}} = \E{|I\cap Z|} + \PR{|I\cap Z|=0} &= (1-\rho)L' + \frac{\binom{L+1-s}{L'}}{\binom{L+1}{L'}}\\&\le (1-\rho)L' + \inparen{\frac{d}{L+1}}^s
			\\&=  (1-\rho)L' + \inparen{\frac{d}{L+1}}^{(1-\rho)(L+1)} \eperiod
		\end{align*}
		We claim that $\inparen{\frac{d}{(L+1)}}^{(1-\rho)(L+1)}\le 1-\frac{L'\cdot \rho}{L+1}$ for all $0\le \rho \le \frac{L}{L+1}$. Since the left-hand side is convex in $\rho$, it suffices to verify the claim for $\rho = 0$ and $\rho = \frac{L}{L+1}$, both of which are straightforward to check. 
		We conclude that
		$$
		\sum_{i}\dim \phi(\cV_i) \ge n\cdot\inparen{L' + 1 - (1-\rho)L'- \inparen{1 - \frac{L'\cdot \rho}L}} = n\cdot \rho\cdot L'\cdot \inparen{1+\frac 1L}\eperiod§
		$$
		By the above and \cref{eq:degVUR},
		$$\pdeg\inparen{\cV,U,R} \le n\cdot\inparen{\frac{(L+1)(1-\rho)-1}L} = \pdeg\inparen{\cV,\F_q^{L+1},R}\eperiod$$
		\cref{eq:RLCLDNegativeBoundNTS} follows.

        \end{proof}

    \subsection{A Negative Bound for List-Recoverability}
        Our goal in this section is to prove the following theorem.
    \RLCLR*

        \begin{proof}
        Let $\alpha := 1-\frac{L+1}{L+1-\ell}\cdot \rho$ and $R=\left(1- \frac{1}{m}\right)\cdot \alpha$.
        We construct a profile $\cV\in \cL(\F_q^{L+1})^n$ that belongs to the local profile family associated with $(\rho, \ell, L)$-list-recoverability. We then prove that $R_{\cV} \le R$. \cref{eq:R_PDef} then yields $R_{\cP}\le R$. This suffices to prove the theorem, since, clearly, $R_{\cP'} \leq  R_\cP$.
        
        We turn to describe the profile $\cV$. Given $x\in [L+1]$ and $i\in [m]$, let $(x)_i$ denote the $i$-th symbol in the $\ell$-ary representation of $x$. For $i\in [m]$, let 
        $$E_i = \inset{u\in \F_q^{L+1} \mid \forall x,y\in [L+1]~~(x)_i=(y)_i \implies u_x = u_y}\eperiod$$

        In the profile $\cV$, each of the subspaces $E_j$ ($j\in [L+1]$) appears $\frac{n\cdot\alpha}m$ times. The other $(1-\alpha)\cdot n$ entries of $\cV$ are $\F_q^{L+1}$. Let $G = \inset{i\in [n]\mid \cV_i = \F_q^{L+1}}$.
        
        Observe that a code containing $\cV$ is not $\LR \rho \ell L$. Indeed, let $A\in \cM_{\cV}^{\distinct}$. We claim that the columns of $A$ are $\RC \rho \ell$. This is straightforward to verify given the input lists $Z_1,\dots, Z_n$ that we now define. For $i\in [n]\setminus G$ take $Z_i$ to be the set of all letters in $\mrow Ai$ and observe that this is a set of size at most $\ell$. Next, fix sets $P_i$ for all $i\in G$ such that when $i$ is uniformly sampled from $G$, the set $P_i$ is uniform over $\binom{[L+1]}\ell$. 
        For $i\in G$, let $Z_i = \inset{A_{i,j}\mid j\in P_i}$.

        Let $U$ denote the linear space $\sum_{i=1}^m E_i$. 
        Clearly, $U\in \cL^\distinct(\F_q^{L+1})$, since every $x\ne y\in [L+1]$ have some $i\in [m]$ for which $(x)_i\ne (y)_i$. We claim that $\pdeg(\cV,U,R)\ge \pdeg(\cV,W,R)$ for every linear subspace $W\subseteq U$. 
        By \cref{eq:RVUDefAlt,eq:RVDef}, this implies that $R_\cV \leq R_{\cV,U} \leq R$, and the lemma follows.        
        
        To proceed, we need the observation that 
        \begin{equation}\label{eq:EsimIntersection}
            E_j\cap \inparen{\sum_{j'\in [m]\setminus\{j\}} E_{j'}} = \spn\{\1\}
        \end{equation}
        for all $j\in [m]$. Indeed, let $u\in E_j$ and suppose that $u \notin \spn\{\1\}$. Then, there exist $x,y\in [L+1]$ such that $u_x\ne u_y$. Let $z\in [L+1]$ such that $(z)_j = (y)_j$, and $(z)_r = (x)_r$ for all $r\in [m]\setminus \{j\}$. Note that $u_z = u_y$ since $z\sim_j y$. Hence, $u_z\ne u_x$. But, $z\sim_{r} x$ for all $r\ne j$, so $v_x = v_y$ for all $v\in \sum_{j'\in [m]\setminus\{j\}} E_{\sim_{j'}}$. Therefore, $u\notin \sum_{j'\in [m]\setminus\{j\}} E_{\sim_{j'}}$. 

        By \cref{eq:EsimIntersection}, $\dim U = (\ell-1) m+1$. Hence,
        \begin{align*}\pdeg(\cV,U,R) &= \sum_{i\in [n]\setminus G} \dim(V_i\cap U) + \sum_{i\in G} \dim(V_i\cap U) - (1-R)\cdot n\cdot \dim U
        \\&=
        \alpha\cdot n\cdot \ell + (1-\alpha)\cdot n\cdot ((\ell-1) m+1) - (1-R)\cdot n\cdot ((\ell-1) m+1) \\
        &= n\cdot \inparen{((\ell-1) m+1)R - (\ell-1)(m-1)\alpha} = n\cdot R\eperiod        
        \end{align*}

        Let $W\subseteq U$ be a linear subspace that maximizes $\pdeg(\cV,W,R)$. We may assume that $\1\in W$ since otherwise 
        $$\pdeg(\cV,W+\spn\{\1\},R) - \pdeg(\cV,W,R) \ge \sum_{i=1}^n \dim(\spn\{\1\}\cap \cV_i) - (1-R)n = n- (1-R)n = R \ge 0\eperiod$$ 
        
        Let $W'\subseteq W$ be a subspace such that $W = \spn\{\1\} + W'$. Similarly, for $i\in [m]$, let $E_i'\subseteq E_i$ be a subspace for which $E_i = \spn\{\1\} + E_i'$. Now,
        \begin{align}
            \sum_{i\in[m]}\dim(E_i\cap W) &\leq  \sum_{i\in[m]}\left( \dim E_i +\dim W - \dim (E_i+W) \right)   \notag\\
            &= \sum_{i\in[m]}\left(2+ \dim E_i' +\dim W'- \dim (E_i+W) \right)   \notag\\
            &= \sum_{i\in[m]}\left(2+ \dim E_i' +\dim W'- \dim (\mathrm{span}\{\vec 1\} + E_i'+W') \right)   \notag\\
            &= \sum_{i\in[m]}\left(1+ \dim E_i' +\dim W'- \dim ( E_i'+W') \right)   \notag\\
            &= \sum_{i\in[m]}\left(1+ \dim ( E_i'\cap W') \right)   \notag\\
            &\leq m + \dim W' = m+ \dim W -1 \notag\ecomma
        \end{align}
        where the last inequality is due to \cref{eq:EsimIntersection}. Therefore, taking $D := \dim W$, we have
        \begin{align*}\pdeg(\cV,W,R) &= \sum_{i\in [n]\setminus G} \dim(V_i\cap W) + \sum_{i\in G} \dim(V_i\cap W) - (1-R)\cdot D
        \\&\le
        \alpha\cdot n\cdot \frac{m+D - 1}m + (1-\alpha)\cdot n\cdot D - (1-R)\cdot D \\
        &= n\cdot \inparen{DR - \frac{(D-1)(m-1)}m\cdot \alpha } = n\cdot R\le \pdeg(\cV,U,R)\eperiod        
        \end{align*}
        
    \end{proof}

	\section{Random RS Codes and RLCs are locally equivalent}\label{sec:RSLDL}
	
	In this section we prove our main theorem about random RS codes, restated below.
	\reduction*

	The technical core of the proof of \cref{thm:MainRSIntro} is the following proposition.
	
	\begin{proposition}\label{prop:RSContainsProfile}
		Let $n\le q$ with $q$ a prime power, and let $b\in \N$. Let $\cV\in \inparen{\cL\inparen{\F_q^b}}^n$ be a $b$-local profile. Let $0 \le R\le 1$ and $\eps > 0$ such that
		\begin{equation}\label{eq:RsContainsProfileRDef}
			\pdeg(\cV,U,R)  \le -\eps \cdot \dim U \cdot n \end{equation}
		for all $U\in \cL\inparen{\F_q^b}$, $U\ne \{0\}$. Let $\cC = \CRS[\F_q]{\alpha_1,\dots,\alpha_n}k$ where $k = Rn$ and $\alpha_1,\dots,\alpha_n$ are identically and independently sampled uniformly from $\F_q$.
		Then,
		\begin{equation}\label{eq:MainPropBadEvent}
		    \PR{\cC\textrm{ contains a nonzero $n\times b$ matrix satisfying } \cV} \le (2^b-1) \cdot \inparen{\frac {(4b)^{4b}k}{\eps q}}^{\frac{\eps n}{2b}}\eperiod 
		\end{equation}
	\end{proposition}
	
	We defer proving \cref{prop:RSContainsProfile} to \cref{sec:RSContainsPRofile}. We now show how this proposition implies \cref{thm:MainRSIntro}.
	
	\begin{proof}[Proof of \cref{thm:MainRSIntro} given \cref{prop:RSContainsProfile}]
\phantom{a}
        \subsubsection*{Reduction from RLCs to Random RS codes}
        We first prove \cref{eq:MainReductionRLCtoRS}. Fix $\cV = (\cV_1,\dots, \cV_n) \in \cF$. By \cref{lem:RVAlt}, there is some $W\in \cL\inparen{\F_q^b}\setminus \cL_{\distinct}\inparen{\F_q^b}$ such that $\pdeg(\cV,W,R_{\cV}) \ge \pdeg(\cV,U,R_{\cV})$ for all $U\in \cL(\F_q^b)$. 

		Let $D = \dim W$ and fix a linear map $\varphi:\F_q^b\to \F_q^{b-D}$ with $\ker \varphi = W$. Consider the $(b-D)$-local profile $\cV' = \inparen{\varphi(\cV_1),\dots,\varphi(\cV_n)}$. Note that for $\cC$ to contain $\cV$ it must also contain $\cV'$. Indeed, suppose that $A\subseteq \F_q^{n\times b}$ is a matrix in $\cC$ satisfying $\cV$. Consider the matrix $A'\in \F_q^{n\times (b-D)}$ whose rows are $\varphi(\mrow A1),\dots,\varphi(\mrow An)$. Then it is straightforward to verify that $A'$ satisfies $\cV'$. Furthermore, by linearity, $A'\subseteq \cC$.
		
		We now claim that \cref{prop:RSContainsProfile} can be applied to $\cV'$. Indeed, by the definition of $R_{\cP}$ (\cref{eq:R_PDef}),
		$$
		R\le R_{\cP}-\eps \le R_{\cV}-\eps.
		$$
		Thus, for all $U\in \cL(\F_q^b)$, we have
		$$\pdeg(\cV,U,R) - \pdeg(\cV,W,R) = \pdeg(\cV,U,R_\cV) - \pdeg(\cV,W,R_{\cV}) - \eps n (\dim U - D) \le -\eps n (\dim U-D)$$
		Therefore,
		$$\sum_{i=1}^n \dim(\cV_i\cap U) - \sum_{i=1}^n \dim(\cV_i\cap W) \le (1-R-\eps)\cdot n \cdot (\dim U-D)$$
		
		Now, let $U'\in \cL(\F_q^{b-D})$ such that $U'\ne \{0\}$, and let $U\in \cL(\F_q^b)$ such that $W\subseteq U$ and $\varphi(U) = U'$, that is, $U = \varphi^{-1}(U')$. Note that $\dim U' = \dim U - D$. Then,

            \begin{align*}
            \pdeg(\cV',U',R) &= \sum_{i=1}^n\dim(\cV'_i\cap U') - (1-R)n\cdot \dim U' \\ &= \sum_{i=1}^n\dim(\varphi(\cV_i \cap U)) - (1-R)n\cdot (\dim U - D) \\&= \sum_{i=1}^n\inparen{\dim(\cV_i\cap U) - \dim(\cV_i\cap W)} - (1-R)n\cdot (\dim U - D) \\
            &= \pdeg(\cV,U,R) - \pdeg(\cV,W,R) \le -\eps \cdot \dim U'\cdot n\eperiod
            \end{align*}
            
	Hence, \cref{prop:RSContainsProfile} applies to $\cV'$. Therefore,
	\begin{align*}\PR{\cC\textrm{ satisfies }\cV} &\le \PR{\cC\textrm{ satisfies }\cV'} \\&\le \PR{\cC\textrm{ contains a nonzero $n\times b$ matrix satisfying } \cV'} \\&\le (2^b-1)\cdot \inparen{\frac {(4(b-D))^{4(b-D)}k}{\eps q}}^{\frac{\eps n}{2(b-D)}}\le (2^b-1)\cdot \inparen{\frac {(4b)^{4b}k}{\eps q}}^{\frac{\eps n}{2b}}\eperiod
        \end{align*}
		\cref{eq:MainReductionRLCtoRS} now follows from a union bound on all $\cV\in \cF$.

    \subsubsection*{Reduction from Random RS Codes to RLCs}
    We turn to prove \cref{eq:MainReductionRStoRLC}.
    
    Fix $\cV\in \cF$ such that $R_\cP = R_\cV$.  We assume without loss of generality that $\eps = R - R_\cV$. We shall prove that
    \begin{equation}\label{eq:RSToRLCReductionForV}
        \PR{\cC\text{ contains }\cV} \ge 1 - \binom b2\cdot (2^b-1)\cdot \inparen{\frac {(4b)^{4b}k}{\eps q}}^{\frac{\eps n}{2b}}\ecomma§
    \end{equation}
    which implies \cref{eq:MainReductionRStoRLC}. Note that the lower bound on this probability is exponentially close to $1$, because of the assumption on $q$, and because $b$ is considered to be a constant integer that is much smaller than $n$.
    
    By \cref{lem:RVAlt}, there exist $U\in \cL_\distinct(\F_q^b)$ and $T\in \cL(\F_q^b) \setminus \cL_{\distinct}(\F_q^b)$ such that
    \begin{equation} \label{eq:MainReductionUMaximizer}
        U,T \in \argmax\inset{\pdeg(\cV,*,R_{\cV})}\eperiod
    \end{equation}   
    We may assume that $T\subseteq U$, since \cref{lem:degOfCap} yields
    $$\pdeg(\cV,U\cap T,R_\cV) \ge \pdeg(\cV,U, R_\cV) +  \pdeg(\cV,T, R_\cV) - \pdeg(\cV,U+T, R_\cV)\ge \pdeg(\cV,U, R_\cV)\ecomma$$
    where the second inequality is due to $\cref{eq:MainReductionUMaximizer}$. Hence, $U\cap T\in \argmax\inset{\pdeg(\cV,*,R_{\cV})}$ by virtue of \cref{eq:MainReductionUMaximizer}. 

    Furthermore, we may assume that $U = \F_q^b$ and $T = \{0\}$. We turn to justify this assumption. Write $U = T \oplus S$ for some linear subspace $S\subseteq \F_q^b$ satisfying $\dim S = \dim U - \dim T$. Let $\pi:\F_q^b\to S$ denote the projection map onto $S$ with regard to this direct sum.    
    Let $\phi: \F_q^{\dim S}\to S$ be a linear bijection. Consider the profile $\cV'\in \cL(\F_q^{\dim S})^n$ where $\cV'_i = \phi^{-1}(\pi(\cV_i))$. Let $W' \in \cL(\F_q^{\dim S})$ and set $W = \phi (W')+T$. Then,

    \begin{align*}
        \pdeg(\cV',W',R_\cV) &= \sum_{i=1}^n \dim(\cV'_i\cap W') - (1-R_\cV)\dim W'  \\ 
        &= \sum_{i=1}^n \dim(\pi(\cV_i) \cap \phi (W')) - (1-R_{\cV})(\dim W - \dim T) \\
        &= \sum_{i=1}^n \dim(\pi(\cV_i) \cap W) - (1-R_{\cV})(\dim W - \dim T)\\
        &= \sum_{i=1}^n \inparen{\dim(\cV_i \cap W) - \dim(\cV_i \cap T)} - (1-R_{\cV})(\dim W - \dim T) \\
        &= \pdeg(\cV,W,R_{\cV})- \pdeg(\cV,T,R_{\cV})\eperiod
    \end{align*}

    The penultimate equality follows from the rank nullity theorem applied to the linear map $\pi$ and the vector space $\cV_i \cap W$, and the observation that $\pi(\cV_i \cap W) = \pi(\cV_i) \cap W$.
    
    Observe that $\pdeg(\cV,T,R_\cV)$ does not depend on $W'$. Therefore, $\pdeg(\cV',W',R_V)$ is maximized when $W'$ corresponds to a space $W\in \argmax\inset{\pdeg(\cV,*,R_{\cV})}$. In particular $\F_q^{\dim S}$ and $\{0\}$ correspond to $U$ and $T$, respectively. Thus, by \cref{eq:MainReductionUMaximizer}, $\F_q^{\dim S},\{0\}\in \argmax\inset{\pdeg(\cV',*,R_{\cV})}$ and so by \cref{cor:RVAltContrapos}, $R_{\cV'} = R_{\cV}$. If \cref{eq:RSToRLCReductionForV} applies to $\cV'$ then it also applies to $\cV$ since, by linearity of $\cC$, containing $\cV$ implies containing $\cV'$. It therefore suffices to prove \cref{eq:RSToRLCReductionForV} under the assumption that $U= \F_q^b$ and $T = \{0\}$. We proceed to do so.

    Denote $K = \cC^b\cap \cM_{\cV}$. 
    By \cref{eq:MainReductionUMaximizer},
    \begin{align*}
        \pdeg(\cV,\F_q^b,R) = \pdeg(\cV,\F_q^b,R_\cV) + \eps \cdot n\cdot b 
        = \pdeg(\cV,\{0\},R_\cV) + \eps \cdot n\cdot b = \eps \cdot n\cdot b\eperiod
    \end{align*}
    Thus, by \cref{eq:dimM*DegBound}, it holds deterministically that
    $$
    \dim K \ge \pdeg(\cV,\F_q^b,R) = d\ecomma
    $$
    where $d = \ceil{\eps\cdot n\cdot b}$.

    Denote $U_{u,v} = \inset{x \in \F_q^b \mid x_u=x_v}$ and $K_{u,v} = \cC^b\cap \cM^*_{\cV,U_{u,v}}$. Let $E_{u,v}$ denote the event that $\dim K_{u,v} \ge d$. We claim that if none of the events $E_{u,v}$ hold then $\cC$ contains $\cV$. Indeed, assume that no event $E_{u,v}$ holds. Note that $$\cC^b\cap M_\cV^\distinct = K\cap \cM_\cV^\distinct = K \setminus \inparen{\bigcup_{u,v} K_{u,v}}\eperiod$$
    Hence,
    $$\inabs{\cC^b\cap M_\cV^\distinct} \ge \inabs{K} - \sum_{u,v}\inabs{K_{u,v}} \ge q^d - \binom b2 \cdot q^{d-1} > 0\ecomma$$
    where the last inequality is by our assumption that $q > kb\ge \eps n b \ge 2b^2(b+1)$. Consequently, $$\cC^b\cap M_\cV^\distinct \ne \emptyset\ecomma$$ so $\cC$ contains $\cV$ and thus satisfies $\cP$.

    We turn to bound the probability of $E_{u,v}$ for some fixed $u,v$. To this end, we find a necessary condition for this event. Let $\cD = \CRS[\F_q]{\alpha_1,\dots,\alpha_n}{k-t}$,  
    where $t = \floor{\frac{d-1}{b-1}}$. 
    Observe that $\cD$ is a linear subspace of $\cC$ of codimension $t$. Moreover, $\cD^b \cap \F_q^{n \times U_{u,v}}$ is a subspace of $\cC^b \cap \F_q^{n \times U_{u,v}}$, with codimension $t\cdot \dim U_{u,v} = t\cdot(b-1)$. Therefore,
    $$\dim \inparen{\cD^b\cap \cM^*_{\cV,U_{u,v}}} \ge \dim \inparen{\cC^b\cap \cM^*_{\cV,U_{u,v}}} - t(b-1) = \dim K_{u,v} - t(b-1)\eperiod$$
    Therefore, the event $E_{u,v}$ implies that $\dim \inparen{\cD^b\cap \cM^*_{\cV,U_{u,v}}} \ge 1$. We next bound the probability of the latter.

    Let $\psi:\F_q^b\to \F_q^{b-1}$ be the projection map onto the coordinate set $[b]\setminus \{u\}$. Define the $(b-1)$-local profile $\cT\in \cL(\F_q^{b-1})^n$ by $\cT_i = \psi(\cV_i\cap U_{u,v})$. Note that $\dim\inparen{\cD^b\cap \cM^*_{\cV,U_{u,v}}} = \dim\inparen{\cD^{b-1}\cap \cM_{\cT}}$. Moreover, let $W'\in \cL(\F_q^{b-1})$ and let $W\in \cL(\F_q^b)$ such that $W\subseteq U_{u,v}$ and $\psi(W) = W'$. Then, denoting $\delta = \frac tn$, we have
    \begin{align*}
    \pdeg(\cT,W',R_\cV-\delta) &= \sum_{i=1}^n \dim(\cT_i \cap W') - \inparen{1-(R_\cV-\delta)}\cdot n\cdot \dim W' 
    \\&= \sum_{i=1}^n \dim(\cV_i \cap W) - \inparen{1-(R_\cV-\delta)}\cdot n\cdot \dim W \\
    &= \pdeg(\cV,W,R_\cV-\delta) = \pdeg(\cV,W,R_\cV) - \delta \cdot n\cdot \dim W \\ &\le \pdeg(\cV,\{0\},R_\cV) - \delta \cdot n\cdot \dim W  
    \le - \delta \cdot n\cdot \dim W \ecomma
    \end{align*}
    where the first inequality is due to \cref{eq:MainReductionUMaximizer}. Hence, \cref{prop:RSContainsProfile} yields
    \begin{align*}
        \PR{E_{u,v}} \le \PR{\dim \inparen{\cD^b\cap \cM^*_{\cV,U_{u,v}}} \ge 1} = \PR{\dim \inparen{\cD^b\cap \cM_{\cT}} \ge 1} &\le (2^b-1) \cdot \inparen{\frac {(4b)^{4b}(k-t)}{\delta q}}^{\frac{\delta n}{2b}} \\ &\le 
        (2^b-1)\cdot \inparen{\frac {(4b)^{4b}k}{\eps q}}^{\frac{\eps n}{2b}}\eperiod
    \end{align*}
    \cref{eq:RSToRLCReductionForV} follows since
    $$\PR{\cC\textrm{ contains }\cV} \ge 1- \PR{\bigcup_{u,v} E_{u,v}} \ge 1- \binom b2\cdot (2^b-1)\cdot \inparen{\frac {(4b)^{4b}k}{\eps q}}^{\frac{\eps n}{2b}}\eperiod$$
    \end{proof}
	
	\section{On local profiles in a random RS code---Proof of \cref{prop:RSContainsProfile}}\label{sec:RSContainsPRofile}

    In this section we prove \cref{prop:RSContainsProfile}. Our first step is to extend the notion of a local profile, defining what we term a \deffont{local polynomial profile}.

    \subsection{Local Polynomial Profiles}
    We require some notation. Fix $a\in \N$. A tuple of polynomials $(P_1,\dots, P_a)\in \F_q[X]^a$ is said to be \deffont{$k$-bounded} if $\deg(P_i)< k$ for every $1\le i\le a$. We denote the set of all $k$-bounded tuples by
    $$
    Q_{k,a} = \inset{\inparen{P_1,\dots,P_a}\in \F_q[X]^a \mid \forall ~1\le i\le a,~~\deg(P_i) < k }\eperiod
    $$
    Observe that $Q_{k,a}$ is an $\F_q$-linear space. If $S\subseteq \F_q[X]^a$ consists of $k$-bounded tuples, we say that $S$ itself is \deffont{$k$-bounded}.
    
    For $\alpha\in \F_q$ and $\bP = \inparen{P_1,\dots, P_a}\in \F_q[X]^a$, we denote the evaluation $\bP(\alpha) := \inparen{P_1(\alpha),\dots, P_a(\alpha)}$. For $S\subseteq \F_q[X]^a$, we denote $S(\alpha) := \inset{\bP(\alpha)\mid \bP\in S}$. If $A\in \F_q[X]^{m\times m'}$ is a matrix of polynomials, let $A(\alpha)\in \F_q^{m\times m'}$ denote the entry-wise evaluation matrix of $A$ on $\alpha$. 

    Let $\F_q(X)$ denote the field of univariate rational functions over $\F_q$. We naturally embed $\F_q$ in $\F_q(X)$ by considering each $\alpha \in\F_q$ as the constant rational function $\alpha \in \F_q(X)$. By abuse of notation, both objects will be denoted by $\alpha$. Note that every $\F_q(X)$-linear space is also $F_q$-linear. Throughout this section we use $\dims $ and $\diml $ to refer to linear dimension of $\F_q$ and $\F_q(X)$, respectively. We similarly distinguish between $\spns$ and $\spnl$, and between $\ranks$ and $\rankl$. Observe that the space $Q_{k,a}$ is $\F_q$-linear but not $\F_q(X)$-linear.

    The following fact will be useful.
    \begin{lemma}[Dimension of $S(\alpha)$]\label{lem:EvalDimProb}
		Let $S\subseteq \F_q[X]^b$ be an $\F_q$-linear space and write  $D = \diml \spnl S$. The following now holds:
		\begin{enumerate}
			\item $\dims  S(\alpha) \le D$ for every $\alpha\in \F_q$.
			\item If $S$ is $k$-bounded ($k\in \N$) then $$\PROver{\alpha \sim \uniform(\F_q)}{\dims  S(\alpha) = D} \ge 1-\frac{Dk}{q}\eperiod$$
		\end{enumerate}
    \end{lemma}

    \begin{proof}
            Let $A\in \F_q[X]^{|S|\times b}$ be a matrix (possibly with infinitely many rows) whose rows are the elements of $S$. Denote $t := \dims S(\alpha) = \ranks  A(\alpha)$. Note that $A$ contains a minor $M \in \F_q[X]^{t\times t}$ such that $\ranks  M(\alpha) = t$. Thus, 
            $$\inparen{\det M} (\alpha) = \det \inparen{M(\alpha)} \ne 0\ecomma$$
            implying that $\det M$ is a non-zero polynomial, so $M$ has full rank over $\F_q(X)$. Thus, $$D = \rankl  A \ge \rankl  M = t = \dims  S(\alpha)\eperiod$$
            The first claim follows.

            We turn to the second claim. Suppose that $S$ is $k$-bounded. Let $M'\in \F_q[X]^{D\times D}$ be a minor of $A$ with $\rankl  M' = D$, so $\det M' \ne 0$. Since the entries of $M'$ are polynomials of degree at most $k$, the non-zero polynomial $\det M'$ has degree at most $Dk$. Hence, a uniformly random $\alpha \in \F_q$ satisfies $\inparen{\det M'}(\alpha) \ne 0$ with probability at least $1-\frac{Dk}q$. Suppose that this event holds. Then, $$D \ge t = \ranks  A(\alpha) \ge \ranks  M'(\alpha) \ge D\ecomma$$
            so $\dims S(\alpha) = t = D$.
   \end{proof}

    We turn to define the notion of a \deffont{local polynomial profile}.

     \begin{definition}\label{def:polyMap}
            An $\F_q(X)$-linear map $\psi:\F_q(X)^b\to \F_q(X)^a$ is said to be \deffont{polynomial} if $\psi\inparen{\F_q[X]^b} \subseteq \F_q[X]^a$. In other words, $\psi$ is \deffont{polynomial} if it is represented in the standard basis by a matrix with polynomial entries. We use $\deg \psi$ to denote the maximum degree (as a polynomial) of an entry in this matrix.
        \end{definition}
        
	\begin{definition} \label{def:polyProfiles}
		A \deffont{$b$-local polynomial profile} is a sequence $\Psi = \inparen{\psi_1,\dots \psi_n}$ of $\F_q(X)$-linear polynomial maps $\psi_i : \F_q(X)^b \to \F_q(X)^{b_i}$ (for some $b_1,\dots, b_n\in \N$). We denote $\deg \Psi = \max_{1\le i\le n} \deg \psi_i$.
  
        Let $\balpha = \inparen{\alpha_1,\dots,\alpha_n} \in \F_q^n$. A tuple of polynomials $(P_1,\dots, P_b)$ is said to \deffont{satisfy} $(\Psi,\balpha)$ if for every $i\in [n]$ it holds that
		\begin{equation}\label{eq:PsiConstraint}
			\inparen{\psi_i\inparen{P_1,\dots, P_b}}(\alpha_i)=0\eperiod
		\end{equation}
        For an $\F_q$-linear space $S\subseteq \F_q[X]^b$, we denote $$S[\Psi,\balpha] = \inset{\bP\in S\mid \bP\textrm{ satisfies }(\Psi,\balpha)}\eperiod$$
        \end{definition}
        Note that $S[\Psi,\balpha]$ is an $\F_q$-linear space.
        
        Polynomial profiles provide a framework for directly analyzing and manipulating message vectors, which, in the case of random RS codes, correspond to the coefficient vectors of low-degree polynomials. This generalization enables one to define profiles that are well-suited for the stochastic process to proceed until completion (see \cref{sec:LikelyToBeEmpty} for further details). Additionally, we require a notion that is related to the concept of $\pdeg$, which is defined for subspaces of $\F_q^b$.
        
        \begin{definition}[Strength]\label{def:strength}
            For an $\F_q(X)$-linear subspace $U \subseteq \F_q(X)^b$ and a polynomial map $\psi$, we denote $\rf_U(\psi) \coloneqq \diml  \psi(U)$. Note that when $U = \F_q(X)^b$, we have $\rf_U(\psi) = \rankl  \psi$.
            
            For an $\F_q(X)$-linear subspace $U\subseteq\F_q(X)^b$ and a polynomial profile $\Psi = \inparen{\psi_1,\dots \psi_n}$, we denote $\rf_U(\Psi) = \sum_{i=1}^n \rf_U(\psi_i)$.
        \end{definition}

    \begin{definition}
        Given $T,U\subseteq \F_q(X)^b$, where $U$ is an $\F_q(X)$-linear subspace, we say that $U$ is \deffont{$T$-live} if $\spnl (T\cap U) = U$. Otherwise, we say that $U$ is \deffont{$T$-dead}.
    \end{definition}

    Our goal is to rephrase \cref{prop:RSContainsProfile} in terms of local polynomial profiles. We require the following lemma to do so.
    \begin{lemma}\label{lem:LocalProfilesToPolynomialProfiles}
        Let $n\le q$ with $q$ a prime power. Let $\cV =  \inparen{V_1,\dots, V_n}\in \inparen{\cL\inparen{\F_q^b}}^n$ be a $b$-local profile. Then, there exists a $b$-local polynomial profile $\Psi_{\cV}$ with $\deg\inparen{\Psi_{\cV}}=0$ such that, for every $\balpha = \inparen{\alpha_1,\dots,\alpha_n}\in \F_q^n$ and $k\in \N$, the following holds:
        \begin{equation}\label{eq:LocalProfToPolyProfEquivalence}
        \CRS[\F_q]{\alpha_1,\dots,\alpha_n}k\textrm{ contains a nonzero $n\times b$ matrix satisfying } \cV\quad\Rightarrow\quad \dims Q_{k,b}[\Psi_{\cV},\balpha] > 0\eperiod
        \end{equation}
        Furthermore, for every $Q_{k,b}$-live $\F_q(X)$-linear space $U\subseteq \F_q(X)^b$ there exists an $\F_q$-linear space $U'\subseteq \F_q^b$ with $\diml U = \dims  U'$ and
        
        \begin{equation}\label{eq:strengthBoundedByDeg}
        \rf_U(\Psi) \ge \sum_{i=1}^n \inparen{\diml U  - \dims \inparen{U'\cap V_i}}\eperiod\end{equation}
    \end{lemma}

    \begin{proof}
        For $1\le i\le n$, let $\varphi_i: \F_q^b\to \F_q^{b-\dims  V_i}$ be an $F_q$-linear map with $\ker \varphi_i = V_i$. Take $\psi_i : \F_q(X)^{b} \to \F_q(X)^{b-\dims  V_i}$ to be an $\F_q(X)$-linear map whose restriction to $\F_q^b$ is $\varphi_i$. In other words, the matrix representation of $\psi_i$ in the standard basis is the same as that of $\varphi_i$. In particular, $\deg \psi_i = 0$. Observe that $\varphi_i(\bP(\alpha)) = \psi_i(\bP)(\alpha)$ for all $\alpha \in \F_q$ and $\bP\in \F_q[X]^b$.
        Finally, take $\Psi_\cV = \inparen{\psi_1,\dots, \psi_n}$ and note that $\deg \Psi_\cV = 0$.

        Now, suppose that $\cC:=\CRS[\F_q]{\alpha_1,\dots,\alpha_n}k$ contains a non zero matrix $A\in \F_q^{n\times b}$ that satisfies $\cV$. By definition of the RS code, there is a tuple of distinct polynomials $\bz\ne \bP:=\inparen{P_1,\dots, P_b}\in Q_{k,b}$ such that $A_{i,j} = P_j(\alpha_i)$ for all $1\le i\le n$ and $1\le j\le b$. Observe that $\bP$ satisfies $\inparen{\Psi_\cV,\balpha}$. Indeed, 
        $$\psi_i(\bP)(\alpha_i) = \varphi_i\inparen{\bP(\alpha_i)} = \varphi_i(\mrow Ai) = 0\ecomma$$
        since $\mrow Ai\in V_i = \ker \varphi_i$. Therefore, $\bP \in Q_{k,b}[\Psi_\cV,\balpha]$, and thus $\dims Q_{k,b}[\Psi,\balpha] > 0$.

        We turn to proving the last part of the lemma. Let $U\subseteq \F_q(X)^b$ be $Q_{k,b}$-live and denote $D = \diml U$. Let $S = Q_{k,b}\cap U$. Note that $S \subseteq \F_q[X]^b$ is an $\F_q$-linear space and is $k$-bounded, and therefore by the second item of \cref{lem:EvalDimProb}, 
        $$\PROver{\alpha \sim \uniform(\F_q)}{\dims  S(\alpha) = D} \ge 1-\frac{Dk}{q}\ge 1- \frac {bk}q > 0\eperiod$$
        In particular, we may fix some $\alpha \in \F_q$ such that $\dims S(\alpha) = D$. We take $U' =S(\alpha)$. It remains to show that $U'$ satisfies \cref{eq:strengthBoundedByDeg}.

        Fix $1\le i\le n$.  Then,
        \begin{align*}
        \dim(U'\cap V_i) &= D - \dim \varphi_i(U') \\&= D - \dim\inset{\varphi_i(\bP(\alpha))\mid \bP\in S} \\&= D - \dim\inset{\psi_i(\bP)(\alpha)\mid \bP\in S}
        \\&\ge
        D - \diml \inparen{\spnl\psi_i(S)} \\&\ge 
        D - \diml \psi_i(U)\eperiod
        \end{align*}
        The first inequality above follows from the first item of \cref{lem:EvalDimProb}, and the second holds because $U$ is $Q_{k, b}$-live. Hence,
        $$\rf_U(\Psi) = \sum_{i=1}^n \diml \psi_i(U) \ge \sum_{i=1}^n \inparen{D - \dims\inparen{U'\cap V_i}}\eperiod$$

    \end{proof}

    We give some intuition for $\rf$. Recall that for $U, U'$ as defined in the previous lemma, we have $D=\diml U=\dims U'$. By \cref{eq:strengthBoundedByDeg}, we see that
    \begin{align*}
        \pdeg(\cV, U', R) &= \sum_i \dims(\cV_i \cap U') - (1-R)n\dims U' \\
        &= -(\sum_i \dims U'-\dims(\cV_i \cap U)) + Rn\dims U' \\
        &\ge Rn\dims U' - \rf_U(\Psi) \\
        &= Rn\diml U - \rf_U(\Psi)
    \end{align*}
    Thus, we observe a close relationship between $\pdeg$ and $\rf$. The following lemma, which facilitates the transition to working exclusively with $\mathbb{F}_q(X)$-linear structures, makes this connection more explicit.

    \begin{lemma}\label{lem:LikelyToBeEmptyInitial}
        Fix a prime power $q$ and let $k,b\in \N$. Let $\Psi$ be a $b$-local polynomial profile with $\deg(\Psi) = 0$. Suppose that there exists $\lambda \ge 2b(b+1)$ such that for every $Q_{k,b}$-live $\F_q(X)$-linear space $U\subseteq \F_q(X)^b$ with $U\ne \{0\}$, there holds
        \begin{equation}
        \label{eq:LikelyToBeEmptyInitialStrengthGap}
        \rf_U(\Psi) \ge  (k+\lambda)\cdot \diml  U \eperiod
        \end{equation}
        Then, 
        $$\PROver{\balpha}{\dims Q_{k,b}[\Psi,\balpha] > 0} \le (2^b-1) \cdot \inparen{\frac{4e \cdot 2^{b} \cdot b^{2b+1}\cdot k\cdot n}{\lambda\cdot q}}^{\frac{\lambda}{2b}}\ecomma$$
        where $\balpha$ is sampled uniformly from $\F_q^n$.
    \end{lemma}  

   Before proving \cref{lem:LikelyToBeEmptyInitial}, we show that it implies \cref{prop:RSContainsProfile}.

    \begin{proof}[Proof of \cref{prop:RSContainsProfile} given \cref{lem:LikelyToBeEmptyInitial}]
        Let $\Psi_\cV$ be as in \cref{lem:LocalProfilesToPolynomialProfiles}. Take $\lambda = \eps n$.
        Let $U\subseteq \F_q(X)^b$ be a $Q_{k,b}$-live non-trivial $F_q(X)$-linear space and let $U'\subseteq \F_q^b$ be a corresponding $F_q$-linear space satisfying \cref{eq:strengthBoundedByDeg} with $\dims  U' = \diml  U$. Observe that 
        \begin{align*}
        \rf_U(\Psi_\cV) &\ge \sum_{i=1}^n\inparen{\diml  U - \dims \inparen{U'\cap V_i}} \\
        &= n\cdot\inparen{\diml  U - \dims  U'} + k\cdot \dim U' - \pdeg \inparen{\cV,U',R} \\
        &= k\cdot \dim U' - \pdeg \inparen{\cV,U',R} \\
        &\ge (k+\eps n)\cdot \dim U' & \textrm{by \cref{eq:RsContainsProfileRDef}} 
        \\&= (k+\lambda)\cdot \diml U\eperiod
        \end{align*}

        Therefore, by \cref{eq:LocalProfToPolyProfEquivalence,lem:LikelyToBeEmptyInitial},
        \begin{align*}
            \PR{\CRS[\F_q]{\alpha_1,\dots,\alpha_n}k\textrm{ contains a nonzero $n\times b$ matrix satisfying } \cV} &\le \PR{\dims Q_{k,b}(\Psi,\balpha) > 0} \\&\le (2^b-1) \cdot \inparen{\frac{(4b)^{4b}\cdot k\cdot n}{\lambda\cdot q}}^{\frac{\lambda}{2b}} \\ &=
            (2^b-1) \cdot \inparen{\frac{(4b)^{4b}\cdot k}{\eps\cdot q}}^{\frac{\eps n}{2b}}
            \eperiod
        \end{align*}
    \end{proof}
    
     The rest of this section is devoted to proving \cref{lem:LikelyToBeEmptyInitial}. The main ingredient in our proof is a potential function, which we now define.

    \subsection{A Potential Function for Local Polynomial Profiles}
	Let $\Psi = \inparen{\psi_1,\dots,\psi_n}$ be a $b$-local polynomial profile and let $\balpha = \inparen{\alpha_1,\dots, \alpha_n}$ be a uniformly random element of $\F_q^n$ . To prove \cref{lem:LikelyToBeEmptyInitial} we need to bound the probability that $Q_{k, b}[\Psi,\balpha] \supsetneq \inset{\bz}$. To this end, we will tackle a more general question: given a finite $\F_q$-linear space $S\subseteq \F_q[X]^b$, how likely is it that $S[\Psi,\balpha] \supsetneq \inset{\bz}$? To help answer this question, we define a potential function $\pot$. Given an $\F_q(X)$-linear space $W\subseteq \F_q(X)^b$, we denote
	$$d_W(S) := \dims \inparen{S\cap W}$$
	and
	$$\pot_W\inparen{S,\Psi} := d_W(S) - \rf_W(\Psi)\eperiod$$
	
	To motivate our definition of $\pot_W$, consider an iterative stochastic process in which the random evaluation points $\alpha_1,\dots, \alpha_n$ are revealed one by one.  For $0\le i \le n$, denote $\balpha_\ut i = \inparen{\alpha_1,\dots,\alpha_i}$. Also, $\Psi_\ut i = \inparen{\psi_1,\dots, \psi_i}$ and $\Psi_\uf i = \inparen{\psi_{i},\dots, \psi_n}$. Note that $S[\Psi_\ut i,\balpha_\ut i]$ is fully known after the $i$-th step of the stochastic process. Moreover, $$S = S[\Psi_\ut 0,\balpha_\ut 0]\supseteq S[\Psi_\ut 1,\balpha_\ut 1]\supseteq\dots\supseteq S[\Psi_\ut n,\balpha_\ut n] = S[\Psi,\balpha]\eperiod$$
 
    As we show below (\cref{lem:GammaChange}), the decrease in dimension in a single step, namely,\\ $d_W (S[\Psi_\ut{i-1},\balpha_\ut {i-1}]) - d_W( S[\Psi_\ut i,\balpha_\ut i])$, is at most $\diml  \psi_i(W)$. Therefore,
    \begin{align*}
	d_W(S[\Psi,\balpha]) &= d_W(S) - \inparen{\sum_{i=1}^n d_W (S[\Psi_\ut{i-1},\balpha_\ut {i-1}]) - d_W( S[\Psi_\ut i,\balpha_\ut i])} \\&\ge d_W(S) - \inparen{\sum_{i=1}^n \diml \psi_i(W)}  \\&= d_W(S) - \rf_W(\Psi)= \pot_W(S,\Psi))\numberthis \label{eq:POTPositiveBound}
\end{align*}

    Hence, if $\pot_W(S,\Psi)$ is positive then $S[\Psi,\balpha]\cap W$ is guaranteed to have positive dimension, implying, in particular, that $S[\Psi,\balpha] \supsetneq \{0\}$.     

    We turn to the examine the opposite case, namely, when $\pot_W(S,\Psi)$ is very negative (that is, much smaller than zero). Suppose that $\balpha$ is sampled uniformly at random from $\F_q^n$. One may hope to show that $S[\Psi,\balpha]\cap W = \{0\}$ holds in this setting with high probability. However, this is not always true. Indeed, a construction along the lines of \cref{ex:degMarkovNotTight} yields a setting in which $\pot_W(S,\Psi)$ is very negative, but $\pot_{W'}(S,\Psi)$ is positive for some nontrivial $\F_q(X)$-linear subspace $W'\subseteq W$. By \cref{eq:POTPositiveBound}, we then have $S[\Psi,\balpha]\cap W\supseteq S[\Psi,\balpha]\cap W' \supsetneq \{0\}$.

    In \cref{lem:LikelyToNotContainW} we manage to show that a somewhat weaker probabilistic phenomenon does always occur whenever $\pot_W(S,\Psi)$ is very negative. Namely, we show that $W$ is $S[\Psi,\balpha]$-dead with very high probability.
    
    \begin{lemma}[If $\pot_W(S,\Psi)$ is very negative then $W$ is probably $S\inbrak{\Psi,\balpha}$-dead]\label{lem:LikelyToNotContainW}
		Let $S\subseteq \F_q[X]^b$ be a $k$-bounded $\F_q$-linear space ($k\in \N$). Let $\Psi = \inparen{\psi_1,\dots, \psi_n}$ be a $b$-local polynomial profile and let $g = \deg \Psi$. Let $W \subseteq \F_q(X)^b$ be an $\F_q(X)$-linear space and let $D = \diml  W$. Let $L > 0$ and suppose that $\pot_W\inparen{S,\Psi}\le -L$. Then,
		\begin{equation}\label{eq:ProbOfSpannnigW}
			\PROver{\balpha\sim \uniform\inparen{\F_q^n}}{W\textnormal{ is }S[\Psi,\balpha]\textnormal{-live}} \le \binom n {\ceil{\frac{L}D}}\cdot\inparen{\frac{(k+g)D}{q}}^{\ceil{\frac{L}D}}\eperiod
		\end{equation}
		Furthermore, let $0 \le L'\le L$ and let $E_{L'}$ denote the event that at least one of the following occurs:
        \begin{itemize}
              \item $W$ is $S[\Psi,\balpha]$-live \emph {OR}
              \sloppy
              \item Let $0\le s\le n$ denote the minimum index for which $W$ is $S[\Psi_\ut s,\balpha_\ut s]$-dead. Also, $\pot_W\inparen{S\inbrak{\Psi_\ut{s}, \balpha_{\ut{s}}},\Psi_{\uf{s+1}}} \ge -L'$ holds.
        \end{itemize}
        Then, for all $0\le L'\le L$,
		\begin{equation}\label{eq:ProbOfSmallGamma}
			\PROver{\balpha\sim \uniform(\F_q^n)}{E_{L'}} \le \binom n {\ceil{\frac{L-L'}{D}}}\cdot \inparen{\frac{(k+g)D}{q}}^{\ceil{\frac{L-L'}D}}\eperiod
		\end{equation}
	\end{lemma}
    We prove \cref{lem:LikelyToNotContainW} in the following section.

        \subsection{If $\pot_W(S,\Psi)$ is Very Negative Then $W$ is Probably $S[\Psi,\balpha]$-Dead --- Proof of \cref{lem:LikelyToNotContainW}}\label{sec:LikelyToNoContainW}
        
        Let $S\subseteq \F_q[X]^b$ be an $\F_q$-linear space. Recall that for $\alpha \in \F_q$ we denote $S(\alpha) \coloneqq \inset{\bP(\alpha)\mid \bP\in S}$. We require a key observation about the evaluation map $S\to S(\alpha)$ in order to prove \cref{lem:LikelyToNotContainW}. We prove it in the following lemma.

	\begin{lemma}[The change in $\pot_W$ in a single step]\label{lem:GammaChange}
		Let $\Psi = \inparen{\psi_1,\dots, \psi_n}$ be a $b$-local polynomial profile, let $S\subseteq \F_q[X]^b$ be a $k$-bounded $\F_q$-linear space ($k\in \N$). Let $W \subseteq \F_q(X)^b$ be  an $S$-live $\F_q(X)$-linear space.

        Given $\alpha \in \F_q$, write $p = \pot_W(S,\Psi)$ and $p_\alpha = \pot_W(S[\psi_1, \alpha],\Psi_\uf 2)$. Also, let $D = \diml  \psi_1(W)$ and $g = \deg \psi_1$.
        
        Then, $$p\le p_\alpha \le p +D$$ 
        for all $\alpha \in \F_q$. Furthermore,
		$$\PROver{\alpha\sim \uniform(\F_q)}{p_\alpha = p} \ge 1 - \frac{(k+ g)\cdot D} q\eperiod$$
	\end{lemma}
	\begin{proof}
		Denote
            $$t_\alpha := \dims  (W\cap S) - \dims \inparen{W\cap S[\psi_1, \alpha]} = \dims  \inparen{\psi_1\inparen{W\cap S}(\alpha)}\eperiod$$ Observe that
            $$p_\alpha = p - t_\alpha +D\eperiod$$

            Since $\spnl (W\cap S) = W$, we also have $$\diml \spnl (\psi_1(W\cap S)) = \diml (\psi_1(W)) = D\eperiod$$ 
            Observe that every tuple in $\psi_1(W\cap S)$ is $(g+ k)$-bounded.         
            Hence,
            applying \cref{lem:EvalDimProb} to the space $\psi_1\inparen{W\cap S}$ yields $0\le t_\alpha\le D$ for all $\alpha$, and $\PROver{\alpha\sim \uniform(\F_q)}{t_\alpha = D} \ge 1 - \frac{D\cdot (g+ k)}q$.
            
	\end{proof}
	
	We can now prove \cref{lem:LikelyToNotContainW}.

	\begin{proof}[Proof of \cref{lem:LikelyToNotContainW}]
		It suffices to prove \cref{eq:ProbOfSmallGamma}, since the latter implies \cref{eq:ProbOfSpannnigW} by taking $L' = 0$. We turn to prove \cref{eq:ProbOfSmallGamma}. Let $\balpha$ be uniformly sampled from $\F_q^n$. For $0\le i\le n$ denote
        $S_i = S[\Psi_\ut i, \balpha_\ut i]$ and
        $p_i = \pot_W(S_i,\Psi_\uf {i+1})$. Recall that $p_0 = -L$.
        
       Observe that the bad event $E_{L'}$ implies that there is some $1\le t\le n$ such that $W$ is $S_{t-1}$-live and $p_t \ge -L'$. Indeed, suppose that $E_{L'}$ holds and consider two cases. If $W$ is $S_n$-dead, then taking $t=s$ (recall that $s$ is the minimum index for which $W$ is $S_s$-dead) satisfies the claim. Otherwise, take $t=n$. Since $W$ is $S_n$-live it is also $S_{n-1}$-live. Now, $$p_t = p_n = \dims  (W\cap S_n) \ge 0 \ge -L'\ecomma$$
       and the claim holds. Therefore,
        $$\PROver{\balpha}{E_{L'}}\le \PROver{\balpha}{p_t \ge -L' }\eperiod$$

        Let
		$$
		M = \inset{i\in [n] \mid p_{i+1} > p_i}\ecomma
		$$
		namely, $M$ is the set of steps in which the potential function strictly increases. By the first part of \cref{lem:GammaChange}, each step in $M$ increases the potential by at most $D$. Hence, $$p_t \le p_0 + D\cdot \inabs{M \cap [t]} = -L + D\cdot \inabs{M\cap [t]}\eperiod$$
		Thus, a necessary event for $p_t\ge -L'$ is that
		$\inabs{M\cap [t]} \ge m$, where $m = \ceil{\frac{L-L'}D}$. For $0\le i\le n$, let $G_i$ denote the event that $W$ is $S_i$-live (note that $G_i$ implies $G_0,\dots, G_{i-1}$). Note that the events $G_0,\dots, G_{t-1}$ hold by definition of $t$. Thus, taking $I = M\cap [t]$, we have
		\begin{align*}
			\PROver{\balpha}{E_{L'}}\le \PROver{\balpha}{p_t\ge -L'} &\le \PROver{\balpha}{\bigvee_{\substack{I\subseteq [n]\\|I| = m}}\inparen{ I\subseteq M} \wedge G_{\max I-1}} \\
			&\le \sum_{\substack{I\subseteq [n]\\|I| = m}} \PROver{\alpha_1,\dots\alpha_{\max I}}{\inparen{I\subseteq M} \wedge G_{\max I - 1}}\eperiod
		\end{align*}
		The sum on the right-hand side has $\binom{n}{m}$ terms. To prove \cref{eq:ProbOfSmallGamma}, it suffices to show that each of these terms is at most  $\inparen{\frac{(k+g)D}{q}}^{m}$. We proceed to do so.
		
		Fix $I\subseteq [n]$ with $|I| = m$. Write $I = \inset{i_1,\dots, i_m}$ where $i_1 < i_2 <\dots < i_m$. Then,
		\begin{align*}
			\PROver{\alpha_1,\dots\alpha_{i_m}}{I\subseteq M \wedge G_{i_m-1}} &= \PROver{\alpha_1,\dots\alpha_{i_m}}{\inparen{I\subseteq M} \wedge G_0 \wedge G_1\wedge \dots \wedge G_{i_m - 1}} \\
			&\le \prod_{j=1}^m \PROver{\alpha_1,\dots,\alpha_{i_j}}{i_j\in M \mid \inparen{\inset{i_1,\dots,i_{j-1}}\subseteq M} \wedge G_{i_j-1}} \eperiod
		\end{align*}
        The inequality follows from the fact that $G_i$ implies $G_0 \land \ldots \land G_{i-1}$.
		To finish the proof we show that each term on the right-hand side is at most $\frac{(k+g)D}q$. Let $1\le j\le m$.
        
        Condition on $\alpha_1,\dots, \alpha_{i_j-1}$ for which $\inset{i_1,\dots,i_{j-1}}\subseteq M \wedge G_{i_j-1}$ holds. In particular, under this conditioning, $S_{i_j-1}$ is determined and satisfies $\spnl  (S_{i_j-1} \cap W) = W$. Furthermore, recall that $i_j \in M$ is the event where the potential function increases after substituting a random $\alpha_{i_j}$. Thus, for all $\alpha_1,\dots, \alpha_{i_j-1}$ for which $\inset{i_1,\dots,i_{j-1}}\subseteq M \wedge G_{i_j-1}$ holds, by the second part of \cref{lem:GammaChange}, we have
        \[
            \PROver{\alpha_1,\ldots,\alpha_{i_j}}{i_j\in M \mid \alpha_1,\dots, \alpha_{i_j-1}} \le \frac{(k+g)D}{q}.
        \]
        Because the events $\inset{i_1,\dots,i_{j-1}}\subseteq M$ and $G_{i_j-1}$ only depend on $\alpha_1,\dots, \alpha_{i_j-1}$, we see that the $j$-th term in the product above is at most $\frac{(k+g)D}q$.
	\end{proof}

    \subsection{If $\pot_U(S,\Psi)$ is Very Negative for All $U$ Then $S[\Psi,\balpha]$ is Probably Trivial} \label{sec:LikelyToBeEmpty}
        In this section, we prove a stronger version of \cref{lem:LikelyToNotContainW}, which directly implies  \cref{lem:LikelyToBeEmptyInitial}.
            \begin{lemma}[If $\pot_U(S,\Psi)$ is very negative for all $U$ then $S\inbrak{\Psi,\balpha}$ is probably trivial]\label{lem:LikelyToBeEmpty}
		Let $S\subseteq \F_q[X]^b$ be a $k'$-bounded $\F_q$-linear space ($k'\in \N$). Let $W = \spnl  S$ and write $D_W = \diml  W$. Let $\Psi = \inparen{\psi_1,\dots, \psi_n}$ be a $b$-local polynomial profile and denote $g := \max\inset{\deg \Psi, b\cdot k'}$. 
          Suppose that there exists some \begin{equation}\label{eq:lambdaLB}
			\lambda \ge 2\cdot D_W\cdot(D_W+1)
		\end{equation} such that
		\begin{equation} \label{eq:potentialLB}
			\pot_U\inparen{S,\Psi} \le -\lambda \cdot \diml U
		\end{equation}
		for every $S$-live $\F_q(X)$-linear subspace $U\subseteq W$.
        Then,
		$$
		\PR{S\inbrak{\Psi,\balpha} \supsetneq \{0\}} \le  \inparen{2^{D_W}-1}\cdot \inparen{\frac{4e \cdot 2^{D_W} \cdot b^{2D_W}{\cdot g \cdot n}}{\lambda\cdot q}}^{\frac{\lambda}{2D_W}}\eperiod
		$$
	\end{lemma}

        \begin{proof}[Proof of \cref{lem:LikelyToBeEmptyInitial} given \cref{lem:LikelyToBeEmpty}]
            The claim follows by applying \cref{lem:LikelyToBeEmpty} to $S := Q_{k,b}$ and $\Psi$. Note that $g = b\cdot k$. Also, \cref{eq:LikelyToBeEmptyInitialStrengthGap} yields
            $$\pot_U(S,\Psi) = \dims(U\cap Q_{k,b}) - \rf_U(\Psi) \le k\cdot \diml U - \rf_U(\Psi) \le -\lambda \cdot \diml U$$
            for every $Q_{k,b}$-live $U\subseteq W$. The first inequality, $\dims(U\cap Q_{k,b})\le k\cdot \diml U$, follows from the following argument. Let $G\in \F_q(X)^{b\times \diml U}$ be a generating matrix for $U$ in systematic form, namely, $U = \inset{Gx \mid x\in \F_q(X)^{\diml U}}$. Therefore, $G$ contains a $\diml U \times \diml U$ identity matrix as a sub-matrix. Due to the systematic form of $G$, we have
            $$U\cap Q_{k,b} \subseteq \inparen{Gx\mid x\in Q_{k,\diml U}}\eperiod$$
            Thus, $$|U\cap Q_{k,b}|\le |Q_{k,\diml U}| = q^{k\cdot \diml U}\ecomma$$ and so,
            $$\dims (U\cap Q_{k,b}) = \log_q |U\cap Q_{k,b}| \le k\cdot \diml U\eperiod$$
        \end{proof}

        One may attempt to prove \cref{lem:LikelyToBeEmpty} by the following straighforward approach. Denote $S_n = S[\Psi,\balpha]$. An equivalent condition for $S_n= \{0\}$, is that every $\F_q(X)$-linear space $U\subseteq \F_q(X)^b$ is $S_n$-dead.  \cref{eq:lambdaLB,lem:LikelyToNotContainW} imply that this is very likely for any given $U$, so it seems reasonable to attempt a union bound over all $S$-live subspaces $U$ (an analogous union bound over subspaces is inherent to our proof of \cref{prop:RLCThresholdForProfile}). Unfortunately, the number of $S$-live $\F_q(X)$-linear subspaces (roughly $q^{\Omega\inparen{k\cdot b^2}}$) is too large for this approach to work. Therefore, we require a subtler strategy.

        In our proof of \cref{lem:LikelyToBeEmpty}, we follow the run of the iterative stochastic process described at the opening of this subsection. We run the process until a time $s$ in which one of two stopping conditions occurs. The first condition is that the space $W$ (namely, the $\F_q(X)$-span of $S$) is $S_s$-dead, at which point we continue the run with regard to $S_s$, whose $\F_q(X)$-span is now smaller. This is the simple case. 

        The second stopping condition is that some $\F_q(X)$-subspace $U\subseteq W$ has become ``worrisome''. By this, we mean that $p_{U,s} := \pot_{U}(S_s,\Psi_{\uf{s+1}})$ has increased and gotten too close to zero from the negative side. This is indeed worrying since, if this potential becomes positive, it is guaranteed that $S_n \cap U \ne \{0\}$, so we cannot hope to have $S_n = \{0\}$.  In this case, we proceed recursively by analyzing two different runs. First, we consider the set $S_s \cap U$ and show the remaining negative potential $p_{U,t}$ is still negative enough to ensure that $S_n\cap U = \{0\}$ with high probability. Second, we define a ``quotient set'' $S/U$ and ``quotient profile'' $\Psi/U$. Using the fact that $p_{U,s}$ is only slightly negative, we show that this quotient run likely leads to the $\{0\}$-set, meaning that $S_n\subseteq U$ is likely. Together, $S_n\cap U = \{0\}$ and $S_n\subseteq U$ imply that $S_n = \{0\}$.

	\begin{proof}[Proof of \cref{lem:LikelyToBeEmpty}]
        Throughout this proof, if $U$ is an $\F_q(X)$-linear space, we denote $D_U = \diml  U$. We prove the claim by induction on $D_W$.	The case $D_W=0$ is immediate. Suppose that $D_W \ge 1$. Write $L = \lambda\cdot D_W$ and let $L' = \lambda^*\cdot D_W$, where \begin{equation}\label{eq:lambda'Def}\lambda^* = \ceil*{\frac{(D_W-1)\cdot \lambda}{D_W}+D_W}\eperiod\end{equation}
        Let $\balpha$ be sampled uniformly from $\F_q^n$. For $0\le i\le n$, denote $S_i = S[\Psi_{\ut i}, \balpha_{\ut i}]$ and $p_{U,i} = \pot_U(S_i, \Psi_{\uf {i+1}})$ (where $U\subseteq W$ is $F_q(X)$-linear).

        Let $1\le t$ denote the minimal index at which $W$ is $S_t$-dead (if such an index does not exist, let $t = \infty$). Let $F$ denote the event that $t < \infty$ and $p_{W,t} < -L'$.

        Let $1\le s$ denote the minimal index in which at least one of the following two events occur:
        \begin{enumerate}
			\item $W$ is $S_s$-dead OR
			\item For some $S_s$-live $\F_q(X)$-linear space $\{0\}\ne U\subsetneq W$, there holds
			\begin{equation}\label{eq:tCondition}
				p_{U,s} > - \lambda^*\cdot D_U\eperiod
			\end{equation}
		\end{enumerate}
        Let $F'$ denote the event that $s < \infty$ and $p_{W,s} < -L'$. Observe that $F$ implies $F'$. Indeed, assume that $F$ holds. Clearly, $s\le t$, so $s < \infty$. By \cref{lem:GammaChange}, $p_{W,0},\dots, p_{W,n}$ is non-decreasing, so $p_{W,s}\le p_{W,t} < -L'$, implying $F'$. Therefore, by \cref{lem:LikelyToNotContainW},
        \begin{align*}
            \PROver{\balpha}{\overline{F'}} \le \PROver{\balpha}{\overline F} &\le  \binom{n}{\ceil{\frac{L-L'}{D_W}}}\cdot \inparen{\frac{(k'+\deg \Psi)\cdot D_W}q}^{\ceil{\frac{L-L'}{D_W}}} \\
            &\le \binom{n}{\ceil{\frac{L-L'}{D_W}}}\cdot \inparen{\frac{2g\cdot D_W}q}^{\ceil{\frac{L-L'}{D_W}}} 
            \\
            &=  \binom n {\ceil{\lambda - \lambda^*}}\cdot \inparen{\frac{2g\cdot D_W}{q}}^{\ceil{\lambda - \lambda^*}} \\
            &\le \inparen{\frac{2e\cdot n\cdot g\cdot D_W}{(\ceil{\lambda-\lambda^*})q}}^{\ceil{\lambda-\lambda^*}}\\
            &\le \inparen{\frac{4e\cdot n\cdot g\cdot D_W^2}{\lambda q}}^{\ceil{\lambda-\lambda^*}} &\text{by \cref{eq:lambdaLB,eq:lambda'Def}}\\
			&\le \inparen{\frac{4e\cdot n\cdot g\cdot D_W^2}{\lambda q}}^{\frac{\lambda}{D_W}-D_W-1} &\text{by \cref{eq:lambda'Def}} \\
			&\le \inparen{\frac{4e\cdot n\cdot g\cdot D_W^2}{\lambda q}}^{\frac{\lambda}{2D_W}} 
			\eperiod\numberthis \label{eq:PrE}
        \end{align*}

    We will show that, given $F'$, the event $S_n = \{0\}$ is highly likely. This statement, combined with \cref{eq:PrE}, will complete the proof.
    
    Henceforth, we condition on the event $F'$. Note that $\alpha_1,\dots,\alpha_s$ fully determine $F'$ and $s$. Therefore, the evaluation points $\alpha_{s+1},\dots,\alpha_n$ are unaffected by this conditioning and remain uniformly distributed and independent.

       We will assume that $S_s \ne \{0\}$, because otherwise $S_n \subseteq S_s = \{0 \}$, and we are done.
       We select an $S_s$-live $\F_q(X)$-linear subspace $U\subsetneq W$ as follows.
       If $W$ is $S_s$-dead, we set $U = \spnl  S_s$. Otherwise, by our assumption of $F'$ and definition of $s$, there exists an $S_s$-live subspace $\{0\} \ne U \subsetneq W$ with $p_{U,s} > -\lambda^*\cdot D_U$. We then take $U$ to be some inclusion-maximal subspace with this property.
       Note that in either case, $U \ne \{0 \}$ and therefore $\diml U \ge 1$.

       We now consider two events: 
            $S_n\cap U = \{0\}$ and 
            $S_n \subseteq U$.       
       We will prove that each of these events is very likely, and thus, so is their conjunction, namely, $S_n = \{0\}$.
       
        \subsubsection*{Probability of $S_n \cap U = \{0\}$}
        To prove that $S_n \cap U = \{0\}$ is very likely, we apply the induction hypothesis to the set $S' := S_s\cap U$ and the polynomial profile $\Psi' := \Psi_{\uf{s+1}}$. Recall that $U$ is $S_s$-live, which implies $U$ is also $S'$-live, therefore $U = \spnl (S' \cap U) = \spnl S'$ holds. Moreover $D_U < D_W$.

        Since $S$ is $k'$-bounded, so is $S'$. Similarly, $\deg \Psi' \le  g$. 

        Let $V\subseteq U$ be $\F_q(X)$-linear and $S'$-live. Denote $D_V \coloneqq \diml V$. Then,
        
        \begin{align*}
		\pot_V(S',\Psi') &= p_{V,s}\\ &\le p_{V,s-1} + D_V &\text{by the first part of \cref{lem:GammaChange}}\\&\le -\lambda^* \cdot D_V + D_V &\text{by definition of $s$}\\ &= -(\lambda^* -1 )\cdot D_V\eperiod 
        \end{align*}
        Also, $$\lambda^*-1 \ge \frac{(D_W-1)\cdot \lambda}{D_W} \ge 2\cdot(D_W-1)(D_W+1) \ge 2 D_U(D_U+1)\eperiod$$
        Hence, $\lambda^*-1$ satisfies \cref{eq:lambdaLB,eq:potentialLB}.
        Thus, the induction hypothesis is indeed valid here, and yields
       \begin{align*}
		\PROver{\balpha}{S_n\cap U \supsetneq \{0\}\mid F'} &\le \inparen{2^{D_U}-1}\cdot \inparen{ \frac{4e\cdot 2^{D_U} \cdot b^{2D_U} \cdot g\cdot n}{q(\lambda^*-1)}}^{\frac{\lambda^*-1}{2D_U}} \numberthis\label{eq:probSCapUTrivial}
        \end{align*}
        
        \subsubsection*{Probability of $S_n\subseteq U$}
        
	We turn to proving a lower bound on the probability of $S_n\subseteq U$. If $W$ is $S_s$-dead then $U = \spnl  S_s \supseteq S_n$, so the event holds deterministically. We therefore proceed under the assumption that $W$ is $S_s$-live, and that $U$ is an $S_s$-live $\F_q(X)$-linear subspace satisfying $p_{U,s} > -\lambda^* \cdot D_U$.
    
    Our strategy is to use $\Psi$ to construct a new polynomial profile $\Psi'$, such that applying the induction hypothesis to $\Psi'$ implies that $S_n\subseteq U$ with high probability. For each entry $\psi_i$ in $\Psi$, the corresponding entry $\psi'_i$ of $\Psi'$ is a polynomial map chosen so that the following diagram commutes:
    \[
    \begin{tikzcd}
    \mathbb{F}_q(X)^b \arrow[r, "\psi_i"] \arrow[d, "\pi"'] & \mathbb{F}_q(X)^a \arrow[r, "\chi_i"] & \mathbb{F}_q(X)^r \\
    \mathbb{F}_q(X)^{b - d_U} \arrow[rru, "\psi_i'"']
    \end{tikzcd}\eperiod
    \]
    Here, $\pi$ and $\chi_i$ are polynomial maps whose kernels are  $U$ and $\psi_i(U)$, respectively.
    
    In order to apply the induction hypothesis to $\Psi'$, we also need to control the degrees of the polynomial maps $\pi$ and $\psi'$. This construction is achieved by \cref{claim:Cramer,claim:MapToKernel,claim:MapOfPsi,claim:polynomialMapZeroes}, which rely on notions from \cref{def:polyMap}. 
            
            \begin{claim}\label{claim:Cramer}
                Let $\mT\in \F_q[X]^{p\times e}$ ($p \ge e$) be a full-rank (over $\F_q(X)$) matrix of $d$-bounded polynomials ($d\in \N$). Then, there is a full-rank matrix $\mT^*\in \F_q[X]^{(p-e)\times p}$ of $(e\cdot d)$-bounded polynomials such that $\ker \mT^* = \im \mT$.
            \end{claim}
            \begin{proof}
                Suppose without loss of generality that $\mT$ can be written as a block matrix $\inbrak{\begin{array}{c}
					\mT_1\\
					\hline
					\mT_2
			\end{array}}$ where $\mT_1 \in \F_q[X]^{e\times e}$ has full-rank over $\F_q(X)$, and $\mT_2\in \F_q[X]^{(p-e)\times e}$.
			
			Let $\mT^*\in \F_q(X)^{(p-e)\times p}$ be the block matrix $\det(\mT_1)\cdot\inbrak{\begin{array}{c|c}
					-\mT_2\mT_1^{-1}&\mI_{p-e}
			\end{array}}$. Observe that $\ker \mT^* = \im \mT$. By Cramer's Rule, every entry of $\det(\mT_1)\cdot \mT_1^{-1}$ is a polynomial of degree at most $(e-1)\cdot d$. Thus, every entry of $\mT^*$ is a polynomial of degree at most $e\cdot d$.
            \end{proof}
		\begin{claim}\label{claim:MapToKernel}
			There exists an $\F_q(X)$-linear polynomial map $\pi \colon \F_q(X)^b\to \F_q(X)^{b-D_U}$ with $\ker \pi = U$ and $\deg \pi \le D_U\cdot k'$.
   
		\end{claim}
		\begin{proof}
			Since $U$ is $S_s$-live, there is a matrix $\mT\in \F_q(X)^{b\times D_U}$ such that $\im \mT = U$ and the columns of $\mT$ belong to $S_s$. In particular, because $S_s$ is a $k'$-bounded space, every entry of $\mT$ is a polynomial in $\F_q[X]$ of degree at most $k'$. Therefore, $\mT$ satisfies the conditions of \cref{claim:Cramer} with $p=b, e=D_U$, and $d=k'$. Let $\mT^*$ be the matrix guaranteed by \cref{claim:Cramer}. Take $\pi$ to be the $\F_q(X)$-linear map that $\mT^*$ represents in the standard basis. The claim follows immediately.
		\end{proof}

            \begin{claim}\label{claim:MapOfPsi}
                Let $\pi$ be as in $\cref{claim:MapToKernel}$ and let $\psi \colon \F_q(X)^b\to \F_q(X)^a$ ($a\le b$) be an $\F_q(X)$-linear polynomial map of degree $d$ and rank $a$. Then, there exist $r\in \N$ and $\F_q(X)$-linear maps $\psi' \colon \F_q(X)^{b-d_U}\to \F_q(X)^{r}$ and $\chi \colon \F_q(X)^a\to \F_q(X)^{r}$ with the following properties:
                \begin{itemize}
                    \item $\chi\circ \psi = \psi'\circ \pi$.
                    \item $\psi'$ is a polynomial map of degree at most $\max\{b\cdot d,b\cdot D_U\cdot k'\}$
                    \item $\ker \chi = \psi(U)$
                \end{itemize}
            \end{claim}
            \begin{proof}
                In this proof, all dimensions and ranks are with regard to the field $\F_q(X)$.
                Let $\mM_\pi \in \F_q[X]^{(b-D_U)\times b}$ and $\mM_\psi \in \F_q[X]^{a\times b}$ be full-rank matrices representing $\pi$ and $\psi$ in the standard basis, respectively. By \cref{claim:MapToKernel} we know that every entry in $\mM_\pi$ is a polynomial having degree at most $D_U\cdot k$. By hypothesis, we assume that every entry in $\mM_\psi$ is a polynomial having degree at most $d$. By a slight abuse of notation, we will use $\spnl (\mM)$ to denote the $\F_q(X)$-span of the columns of some matrix $\mM$.

                Let 
                \begin{align*}
                r&\coloneqq \diml\inparen{\spnl(\mM_\psi^\intercal) \cap \spnl(\mM_\pi^\intercal)} \\
                &= b - \diml\inparen{\ker \psi + \ker \pi} \\
                &= b - \diml(\ker \psi) - \diml(\ker \pi) + \diml(\ker \psi\cap \ker \pi) \\
                &= a - \inparen{\diml(\ker \pi) + \diml\inparen{\ker \psi\cap \ker \pi}} = a - \inparen{\diml(U) - \diml\inparen{\ker \psi \cap U}} \\ &= a - \diml\inparen{\psi(U)}\eperiod
                \end{align*}
                The first equality follows by the rank-nullity theorem and De-Morgan's law.
                 Let $\mZ \in \F_q[X]^{(b-D_U+a)\times b}$ denote the block matrix 
                 $\inbrak{\begin{array}{c}
					\mM_\pi\\
					\hline
					\mM_\psi
			\end{array}}$.
                Define $e \coloneqq b-D_U+a-r \le b$ and note that $\rankl \mZ = e$. Let $\mT\in \F_q[X]^{(b-D_U+a)\times e}$ be a submatrix of $\mZ$ obtained by selecting a set of $e$ columns so that $\rankl \mT = e$. Let $\mT^*\in \F_q[X]^{r\times (b-D_U+a)}$ be the matrix guaranteed by \cref{claim:Cramer}, corresponding to $\mT$. In particular, $\ker \mT^* = \spnl (\mT) = \spnl (\mZ)$.

                Write $\mT^*$ as a block matrix $\inbrak{\begin{array}{c|c}
					\mP&\mQ
			\end{array}}$
                where $\mP\in \F_q[X]^{r\times (b-D_U)}$ and $\mQ\in \F_q[X]^{r\times a}$. We take $\psi'$ and $\chi$ to be the $\F_q(X)$-linear maps represented by $\mP$ and $-\mQ$, respectively. We turn to show that $\psi'$ and $\chi$ satisfy the required properties.

                The first property is equivalent to $-\mQ\mM_\psi = \mP\mM_\pi$. This is true since $\mT^*\mT = 0$, thus, implying $\mT^*\mZ=0$. Therefore,
                $$0 = \mT^*\mZ = \mP\mM_\pi+\mQ\mM_\psi\eperiod$$
                For the second condition, \cref{claim:Cramer} guarantees that each entry of $\mT^*$, and thus also of $\mP$, has degree at most $\max\{d,D_U\cdot k'\}\cdot e \le \max\{b\cdot d,b\cdot D_U\cdot k'\}$.

                For the third condition, let $u\in U$, and recall that $U = \ker \pi$. Then,
                $$0 = \mT^* \mZ u = \mP\mM_\pi u + \mQ \mM_\psi u = \psi'(\pi(u)) +  \chi(\psi(u)) = \chi(\psi(u))\ecomma$$
                where $\psi'(\pi(u))=0$ as $u \in \ker \pi$, so $\psi(U)\subseteq \ker \chi$. On the other hand, suppose that $x\in \F_q(X)^a$ such that  $\chi(x)=0$. Let $y \coloneqq \bz_{b-D_U}\circ x$ denote the concatenation of $b-D_U$ zeros and $x$. Then, $\mT^*y = \mQ x = \chi(x) = 0$, so $y\in \ker \mT^* = \im \mZ$. Let $w\in \F_q(X)^b$ be such that $\mZ w = y$. Then, $\pi(w) = \mM_\pi w = 0$ and $\psi(w) = \mM_\psi w = x$ hold by the definition of $y$. This implies that $w\in \ker \pi = U$ and $x\in \psi(U)$.
            \end{proof}

            \begin{claim}\label{claim:polynomialMapZeroes}
                Let $\chi \colon \F_q(X)^a\to \F_q(X)^r$ be a polynomial map. Let $\alpha \in \F_q$ and $p\in \F_q[X]^b$ be such that $p(\alpha) = 0$. Then, $\chi(p)(\alpha) = 0$.
            \end{claim}
            \begin{proof}
                It suffices to prove the claim for $r=1$. Write $\chi(p) = \sum_{i=1}^a q_i p_i$ where $q_i \in \F_q[X]$ for each $i$. Then, $$\chi(p)(\alpha) = \sum_{i=1}^a q_i(\alpha)\cdot p_i(\alpha) = \sum_{i=1}^a q_i(\alpha)\cdot 0 = 0\eperiod$$
            \end{proof}

            To lower bound the probability of $S_n\subseteq U$ we apply the induction hypothesis to a space of polynomials and a polynomial profile that we now construct. Let $\pi$ be as in \cref{claim:MapToKernel} and let $S'=\pi(S_s)$. Recall that $\Psi$ is a $b$-local polynomial profile, therefore for every $i$ such that $s+1 \le i\le n$, we have a polynomial map $\psi_i \colon \F_q(X)^b \to \F_q(X)^{a_i}$, for some $0 \le a_i \le b$. By applying \cref{claim:MapOfPsi} for $\pi$ and every $\psi_i$ for $s+1 \le i \le n$, we see that there exists an $r_i$ and $\F_q(X)$-linear maps $\psi_i' \colon \F_q(X)^{b-D_U}\to \F_q(X)^{r_i}$ and $\chi_i \colon \F_q(X)^{a_i}\to \F_q(X)^{r_i}$. Let $\Psi' = \inparen{\psi'_{s+1},\dots, \psi'_n}$. Let $\balpha' = \balpha_{\uf{s+1}}$.

            We claim that 
            \begin{equation}\label{eq:sn-empty-implies-within-u}
                S'[\Psi',\balpha'] = \{0\} \quad \implies \quad S_n\subseteq U\eperiod
            \end{equation}
            Indeed, suppose that $S'[\Psi',\balpha'] = \{0\}$ and let $p \in S_n$. We need to show that $p\in U$. Let $p' = \pi(p)$ and let $i$ be such that $s+1\le i\le n$. Since $p \in S_n$, we have $\psi_i(p)(\alpha_i) = 0$. Therefore,
            $$\psi'_i(p')(\alpha_i) = ((\psi'_i\circ \pi) (p))(\alpha_i) = ((\chi_i\circ \psi_i)(p))(\alpha_i) = \chi_i (\psi_i(p))(\alpha_i) = 0\eperiod$$
            Where the last equality is due to \cref{claim:polynomialMapZeroes}. Therefore, $p' \in S'[\Psi',\balpha']$, whence $\pi(p) = p' = 0$. Thus, $p\in \ker \pi = U$.

            \paragraph{Applying the induction hypothesis: }
            Observe that because $W$ is $S_s$-live and $S_s \subseteq S \subseteq W$, $\spnl S_s = W$ holds. Therefore $\spnl S' = \spnl \pi(S_s) = \pi( \spnl S_s) = \pi(W)$. Recall that $\diml U \ge 1$, let $W' := \spnl  S' = \pi(W)$ and note that $\diml  W' = \diml W - \diml (\ker \pi \cap W) = \diml W - \diml U < \diml W$. We will now apply the induction hypothesis to $S'$, $\Psi'$ and $W'$, after showing that they satisfy the required conditions. First, note that $W'$ satisfies $W' = \spnl S'$ by definition, and we just showed that $\diml W' < \diml W$.

            We now prove that for every $S'$-live $\F_q(X)$-linear space $V'\subseteq W'$, the following holds
            \[
                \pot_{V'}(S',\Psi') \le -\lambda^*\cdot D_{V'} \eperiod
            \]
            This is easily seen to be true when $V'$ is the zero subspace, hence assume that it is not.
            Let $U' = \pi^{-1}(V')$ and note that $U\subset U'\subseteq W$. We now prove the following claim:
            \begin{claim}\label{claim:piIntersection}
                Let $S_s, S', V', U'$ be as defined above. The following holds:
                \begin{equation}
                \pi(S_s \cap U') = \pi(S_s) \cap \pi(U') = S' \cap V'\eperiod
                \end{equation}
            \end{claim}
            \begin{proof}
                The second equality is true by definitions of $S', U'$, so we prove the first equality. The containment $\pi(S_s \cap U') \subseteq \pi(S_s) \cap \pi(U')$ is easy to see. The other containment is proved as follows: fix a $p \in \pi(S_s) \cap \pi(U')$. Because $p \in S'$, there is a $q \in S_s$ so that $p = \pi(q)$. Now $q$ must also be in $U'$, because if not, then $\pi(q) \not\in V'$ by the definition of $U'$, which would contradict the fact that $\pi(q) = p \in V'$. 
            \end{proof}

            Combining \cref{claim:piIntersection} and the following two equations:
            \[
                \ker \pi \cap (S_s \cap U') = U \cap S_s \cap U' = U \cap S_s
            \]
            \[
                \dims (\pi(S_s \cap U')) = \dims (S_s \cap U') - \dims (\ker \pi \cap (S_s \cap U')) \ecomma
            \]
            we see that the following is true
            \begin{equation}\label{eq:smallDimS'V'}
                \dims (S'\cap V') = \dims (S_s\cap U') - \dims (U \cap S_s)\eperiod
            \end{equation}
            Observe that we have used $\dims$ (and not $\diml$) in the previous sentence. Because the spaces and maps we are considering are $\F_q$-linear (in fact, they are $\F_q(X)$-linear) subspaces and maps respectively, we can apply rank-nullity theorem while viewing them as linear over the field $\F_q$.

            Note that for $s+1\le i\le n$, 
            \begin{align}
            \diml (\psi_i'(V')) &= \diml (\psi'_i\circ\pi(U'))) = \diml (\chi_i\circ \psi_i(U')) \nonumber\\&= \diml (\psi_i(U')) - \diml (\psi_i(U')\cap \ker \chi_i) = \diml (\psi_i(U')) - \diml (\psi_i(U)) \label{eq:dimOfPsiPrimeVPrime}
            \end{align}
            where we used the fact that $\ker \chi_i = \psi_i(U)$ for the last equality.
            Now, by \cref{eq:smallDimS'V'} and \cref{eq:dimOfPsiPrimeVPrime} we have
            \begin{align*}
                \pot_{V'}(S',\Psi') &= \dims(S' \cap V') - \rf_{V'}(\Psi') \\
                &= \dims(S_s\cap U')- \dim(U \cap S_s) - \rf_{U'}(\Psi_{\uf {s+1}}) + \rf_{U}(\Psi_{\uf {s+1}}) \\
                &= p_{U',s} - p_{U,s}\eperiod
            \end{align*}

            We will now prove
            \[
                p_{U',s} \le -\lambda^* D_{U'}\eperiod
            \]
            Observe that because $U$ is an inclusion-maximal $S_s$-live subspace such that $p_{U,s} > -\lambda^* \cdot D_U$, it is suffices to prove that $U'$ is $S_s$-live. Note that $\spnl(S' \cap V') = V'$ because $V'$ is $S'$-live, and by \cref{claim:piIntersection} and $\F_q(X)$-linearity of $\pi$, we see that
            \begin{equation}\label{eq:V'EqualsPiOfSsU}
                V' = \spnl(S' \cap V') = \spnl(\pi(S_s \cap U')) = \pi(\spnl(S_s \cap U')) \eperiod
            \end{equation}
            Because $U$ is $S_s$-live, $\spnl(S_s \cap U')$ contains $\ker \pi = U$ and therefore we have that $\pi^{-1}(\pi(\spnl(S_s \cap U'))) = \spnl(S_s \cap U')$. This together with \cref{eq:V'EqualsPiOfSsU} implies $\spnl(S_s \cap U') = \pi^{-1}(V') = U'$, and so $U'$ is $S_s$-live.

            Therefore,
            $$\pot_{V'}(S',\Psi') \le -\lambda^* (D_{U'} - D_{U}) = -\lambda^*\cdot D_{V'}\eperiod$$

            The last equality holds by the rank-nullity theorem applied on the subspace $U'$ and map $\pi$. By \cref{eq:lambdaLB} applied to $\lambda$, and \cref{eq:lambda'Def},
            $$\lambda^* \ge \frac{(D_W-1)\cdot \lambda}{D_W}\ge \frac{2 D_W(D_W+1)(D_W-1)}{D_W} = 2(D_W+1)(D_W-1)\ge 2(D_{W'}+1)\cdot D_{W'},$$
            so $\lambda^*$ satisfies \cref{eq:lambdaLB} with regard to $D_{W'}$.
            Recall that $g$ was equal to $\max(\deg \Psi, b \cdot k')$, and set $\hat{k} \coloneqq b\cdot D_U \cdot k'$ and $\hat{g} \coloneqq b^2g$ . By \cref{claim:Cramer} and  \cref{claim:MapToKernel}, $S'$ is $\hat{k}$-bounded. By \cref{claim:MapOfPsi}, 
            \begin{align*}\deg \Psi'\le \max\inset{b\cdot \deg \Psi, b\cdot D_U\cdot k'} \le bg\eperiod
            \end{align*}
            Note that $\hat{g} \ge \max\inset{\deg \Psi', b\cdot \hat{k}}$.
            
            We have now proven that $S', \Psi',$ and $W'$ indeed satisfy the required conditions, with $\hat{g}, \lambda^*, \hat{k}$ being the new values for $g, \lambda, k'$ respectively, and therefore by \cref{eq:sn-empty-implies-within-u},
            
            \begin{align*}
                \PROver{\balpha}{S_n\not\subseteq U\mid F'} &\le \PROver{\balpha}{S'[\Psi',\balpha']\supsetneq \{0\}\mid F'} 
                \\&\le \inparen{2^{D_{W'}}-1}\cdot \inparen{ \frac{4e\cdot 2^{D_{W'}} \cdot b^{2D_{W'}}\cdot \hat{g}\cdot (n-s)}{q\cdot \lambda^*}}^{\frac{\lambda^*}{2D_{W'}}} 
                \\&= \inparen{2^{D_{W'}}-1}\cdot \inparen{ \frac{4e \cdot 2^{D_{W'}}\cdot b^{2D_{W'}}\cdot b^2\cdot g\cdot (n-s)}{q\cdot\lambda^*}}^{\frac{\lambda^*}{2D_{W'}}}
                \numberthis\label{eq:probSnInU}
            \end{align*}

        \subsubsection*{Putting it All Together - Probability of $S_n = \{0\}$}
        \cref{eq:probSCapUTrivial,eq:probSnInU} yield
        \begin{align*}
            \PROver{\balpha}{S_n\ne \{0\}\mid F'} &\le \PROver{\balpha}{S_n\cap U \ne \{0\}\mid F'} + \\
            &\phantom{{}={}}\PROver{\balpha}{S_n\not\subseteq U \mid F'} 
            \\
            &\le  \inparen{2^{D_U}-1}\cdot \inparen{ \frac{4e \cdot 2^{D_U} \cdot b^{2D_U} \cdot g\cdot n}{q(\lambda^*-1)}}^{\frac{\lambda^*-1}{2D_U}}  + \\ 
            &\phantom{{}={}}\inparen{2^{D_{W'}}-1}\cdot \inparen{ \frac{4e \cdot 2^{D_{W'}} \cdot b^{2(D_{W'}+1)} \cdot g\cdot (n-s)}{q\cdot\lambda^*}}^{\frac{\lambda^*}{2D_{W'}}}
            \\
            &\le \inparen{2^{D_W}-2}\cdot \inparen{ \frac{4e\cdot 2^{D_W-1} \cdot b^{2D_W} \cdot g\cdot n}{q(\lambda^*-1)}}^{\frac{\lambda^*-1}{2(D_W-1)}} 
            \\
            &\le \inparen{2^{D_W}-2}\cdot \inparen{ \frac{4e\cdot 2^{D_W} \cdot b^{2D_W} \cdot g\cdot n}{q\cdot \lambda}}^{\frac{\lambda}{2D_W}}
        \end{align*}
        The second last inequality is true because $D_{W'}<D_W$, $D_U<D_W$ are true, and the last inequality is true by \cref{eq:lambda'Def} and the fact that $D_W>1$ holds (recall that $U \subsetneq W$ and $U$ is a non-trivial subspace).

        Recall that \cref{eq:PrE} yields an upper bound on $\PROver{\balpha}{\overline{F'}}$. Therefore,
        \begin{align*}
            \PROver{\balpha}{S_n\ne \{0\}} &\le \PROver{\balpha}{\overline{F'} \land S_n\ne \{0\}} + \PROver{\balpha}{ S_n\ne \{0\} \land F'} \\
            &\le \PROver{\balpha}{\overline{F'}} + \PROver{\balpha}{F'}\cdot \PROver{\balpha}{ S_n\ne \{0\}  \mid F'} \\
            &\le \PROver{\balpha}{\overline{F'}} + \PROver{\balpha}{S_n\ne \{0\}\mid F'} \\ &\le
            \inparen{\frac{4e\cdot D_W^2\cdot g\cdot n}{q\cdot \lambda}}^{\frac{\lambda}{2D_W}}   + \inparen{2^{D_W}-2}\cdot \inparen{ \frac{4e \cdot 2^{D_W} \cdot b^{2D_W}\cdot g\cdot n}{q\cdot \lambda}}^{\frac{\lambda}{2D_W}} \\
            &\le \inparen{2^{D_W}-1}\cdot \inparen{ \frac{4e \cdot 2^{D_W} \cdot b^{2D_W}\cdot g\cdot n}{q\cdot \lambda}}^{\frac{\lambda}{2D_W}}\ecomma
        \end{align*}
        which yields the lemma.
        \end{proof}
\ifauthors
    \section{Acknowledgments}
	   The authors would like to thank Yeyuan Chen and Zihan Zhang for showing us a counterexample which refuted a claim about list-recovery made in an earlier version of the paper, and also inspired the proof for \cref{thm:RLCThresholdForLRIntro}. The authors would also like to thank Joshua Brakensiek, Yeyuan Chen, Manik Dhar, and Zihan Zhang for pointing out a bug in \cref{sec:RSContainsPRofile} in an earlier version of the paper.
\fi
	\printbibliography
	
	\appendix
	\section{Random Code Models}\label{sec:Models}
	\subsection{Models of Random Linear Codes}
	In this work, a \deffont{random linear code of rate $R$} is the kernel $\cC$ of a uniformly random matrix $P\in \F_q^{(n-k)\times n}$. We note that if $P$ happens not to have full rank, the code $\cC$ would have rate larger than $R$. An arguably more accurate way to interpret the notion of \deffont{random linear code} is to sample a code $\cC'$ uniformly at random from among all linear subspace of $\F_q^n$ of dimension $Rn$. However, the former model is usually nicer to work with because it satisfies the elegant \cref{lem:ProbInRLC}. We justify our use of the former model rather than the latter by the following observation.
	
	\begin{lemma}
		Write $k = Rn$. Let $\cC$ be the kernel of a uniformly random matrix $P\in \F_q^{(n-k)\times n}$. Let $\cC$ be a uniformly random subspace of $\F_q^n$ of dimension $k$. Then, the statistical difference between $\cC$ and $\cC'$ is at most $1-e^{-q^{-Rn}\cdot n}$.
	\end{lemma}
	\begin{proof}
		It is a well known fact that $P$ has full rank with probability  
		$$\prod_{i=k+1}^n\inparen{1-q^{-i}}\ge \inparen{1-q^{-k}}^n\ge e^{-q^{-Rn}\cdot n}\eperiod$$
		The lemma follows since, conditioned on $P$ having full rank, $\cC$ and $\cC'$ have the same distribution.    
	\end{proof}
	
	\subsection{Models of Random RS Codes}
	For us, a random RS code is a code $\cC = \CRS[\F_q]{(\alpha_1,\dots,\alpha_n)}k$ where $\alpha_1,\dots,\alpha_n$ are sampled independently and uniformly from $\F_q$. Consider now the random code $\cC' =  \CRS[\F_q]{(\beta_1,\dots,\beta_n)}k$ where $\beta_1,\dots,\beta_n$ are sampled uniformly from $\F_q^n$ \emph{without repetitions}. While we have chosen to work with the former model out of convenience, we show that our results about list-decodability and list-recoverability essentially apply to the latter model as well. In fact, any result about the \deffont{geometry} of $\cC$ also holds approximately for $\cC'$. This is due to the following lemma, which gives a coupling between the two models such that there exists an approximately distance preserving map between the two, provided that $\frac nq$ is small enough.
	
	\begin{lemma}
		Write $k=Rn$. Then, exists a coupling $(\cC,\cC')$ where $\cC\subseteq \F_q^n$ is a random RS code of dimension $k$ and $\cC'\subseteq \F_q^n$ is a random RS code of dimension $k$ without repetitions, such that there exists a linear bijection $\varphi:\cC\to\cC'$ with
		$$\PR{\max_{x\in \cC}\inset{\wt{x-\varphi(x)}}\ge n-q\inparen{1-e^{-\frac nq}}\cdot(1-\delta)} \le \inparen{\frac{e^{-\delta}}{(1-\delta)^{1-\delta}}}^{q\inparen{1-e^{-\frac nq}}}$$
		for all $0 < \delta < 1$.
	\end{lemma}
	\begin{proof}
		Sample $\alpha_1,\dots, \alpha_n$ uniformly and independently at random from $\F_q$. Let $$I = \inset{i\in [n] \mid \exists j< i~~\alpha_i = \alpha_j}\eperiod$$
		
		Sample $\beta_1,\dots, \beta_n$ as follows: For each $i\in [n]\setminus I$, take $\beta_i = \alpha_i$. Then, sample $\beta_i$ for all $i\in I$ to complete the sequence in a uniform repetition-less manner.
		
		Let $\varphi:\cC\to \cC'$ map $\inparen{Q(\alpha_1),\dots,Q(\alpha_n)}$ to $\inparen{Q(\beta_1),\dots,Q(\beta_n)}$ for every $Q\in \F_q[X]$ with $\deg Q \le k$. Clearly, $\wt{x-\varphi(x)} \le |I|$ for each $x\in \cC$. Hence, to prove the lemma it suffices to give a probabilistic upper bound on $|I|$.
		
		Consider a balls and bin scenario in which $n$ balls are each thrown into one of $q$ bins with uniform probability. Let $X$ denote the number of non-empty bins. Note that $|I|$ is distributed identically to $n-X$. Let $X_1,\dots, X_q$ be variables such that $X_i$ takes $1$ if the $i$-th bin is full, and $0$ if it is empty. Note that $\E{X_i} = 1-\inparen{1-\frac 1q}^n \ge 1-e^{-\frac nq}$, so $\E{X} \ge q \cdot \inparen{1-e^{-\frac nq}}$.
		
		Observe that $X_1,\dots, X_q$ are in \deffont{negative association} (see \cite{Wajc2017}). Indeed, by \cite[Theorem 10]{Wajc2017}, the occupancy numbers, indicating the number of balls in each bin, are in negative association. Since $X_1,\dots,X_q$ are monotone-increasing functions of the occupancy numbers, the relevant closure property \cite[Lemma 9]{Wajc2017}, implies that $X_1,\dots, X_q$ are also negatively associated. Thus, a Chernoff-Hoeffding bound \cite[Theorem 5]{Wajc2017} applies to their sum. Thus,
		$$\PR{|I| \ge n-q\inparen{1-e^{-\frac nq}}\cdot(1-\delta)} = \PR{X \le  q\inparen{1-e^{-\frac nq}}\cdot(1-\delta)} \le \inparen{\frac{e^{-\delta}}{(1-\delta)^{1-\delta}}}^{q\inparen{1-e^{-\frac nq}}}\eperiod$$
	\end{proof}
	
\end{document}